\documentclass[11pt]{article}
\usepackage{amsmath}
\usepackage{graphicx,psfrag,epsf}
\usepackage{enumerate}
\usepackage{natbib}

\newcommand{\blind}{0}

\usepackage{amssymb, amsthm, dsfont, xcolor, tikz, caption, subcaption, booktabs, threeparttable}
\usepackage{mathpazo} 
\definecolor{beamerdefault}{rgb}{0.2,0.2,0.7}
\usepackage[colorlinks=true,
linkcolor=beamerdefault,
anchorcolor=black,
citecolor=beamerdefault,
filecolor=black,
menucolor=black,
runcolor=black,
urlcolor=black]{hyperref}

\theoremstyle{definition}
\newtheorem{example}{Example}
\newtheorem{definition}{Definition}
\newtheorem{assumption}{Assumption}
\newtheorem{lemma}{Lemma}
\newtheorem{proposition}{Proposition}

\begin{document}

\def\spacingset#1{\renewcommand{\baselinestretch}%
{#1}\small\normalsize} 
\spacingset{1.4} 


\if0\blind
{
  \title{An Empirical Framework for Discrete Games with Costly Information Acquisition}
  \author{Youngjae Jeong\thanks{
    Department of Economics, The Ohio State University, Email: {\texttt{jeong.376@osu.edu}}}\hspace{.2cm}
    }
  \date{Draft: \today}
  \maketitle
} \fi

\if1\blind
{
  \bigskip
  \bigskip
  \bigskip
  \begin{center}
    {\LARGE\bf Title}
\end{center}
  \medskip
} \fi

\bigskip
\begin{abstract}
\noindent 
This paper develops a novel econometric framework for static discrete choice games with costly information acquisition.
In traditional discrete games, players are assumed to perfectly know their own payoffs when making decisions, ignoring that information acquisition can be a strategic choice.
In the proposed framework, I relax this assumption by allowing players to face uncertainty about their own payoffs and to optimally choose both the precision of information and their actions, balancing the expected payoffs from precise information against the information cost.
The model provides a unified structure to analyze how information and strategic interactions jointly determine equilibrium outcomes.
The model primitives are point identified, and the identification results are illustrated through Monte Carlo experiments.
The empirical application of the U.S airline entry game shows that the low-cost carriers acquire less precise information about profits and incur lower information costs than other airlines, which is consistent with their business model that focuses on cost efficiency.
The analysis highlights how differences in firms' information strategies can explain observed heterogeneity in market entry behavior and competition.
\end{abstract}

\vfill

\newpage
\spacingset{1.5}
\section{Introduction}
\label{sec:intro}
In many economic environments, decision makers may not perfectly know their own payoffs and acquire information about the payoffs when making choices.
Such information acquisition is often costly and affects players' behavior.
Recent empirical works have documented this phenomenon in various contexts such as insurance choice \citep{brown2024}, laundry detergent choice \citep{joo2023}, and presidential election \citep{liao2024}.
Furthermore, experimental evidence shows that individuals tend to acquire costly information when facing payoff uncertainty \citep{dewan2020, dean2023, almog2024}.
These findings highlight the importance of incorporating costly information acquisition into economic models of decision making under payoff uncertainty.

This paper develops a novel econometric framework for static discrete choice games with costly information acquisition.
Previous discrete choice games with incomplete information (e.g., \citealt{bajari2010, aradillas2010}) typically assume that each player fully observes their own payoff but not their rivals' payoffs.
This paper relaxes this assumption and allows players to face uncertainty about their own payoffs and to acquire information about them.
The players optimally choose how much information to acquire about the payoffs, taking into account the costs, and thus determine the precision of their information.
This generalization provides a more realistic and flexible framework for analyzing strategic interactions under incomplete information.

Existing literature in industrial organization utilizes static discrete game models to analyze the strategic interactions between firms.
These models typically impose strong assumptions about players' information, ignoring the possibility that information acquisition itself can be a strategic choice.
For example, the models used to study firms' entry decisions, location decisions, and coordination \citep{BR1991, magnolfi2022, berry1992, ciliberto2009, koh2023, seim2006, sweeting2009} assume that players have either complete information or incomplete information about their own payoffs.
These assumptions on information overlook the fact that players may not perfectly know their own payoffs and may choose to acquire information strategically.

Unlike the previous discrete games with incomplete information, this paper models the information structure as the outcome of a strategic choice involving costly information acquisition.
Each player's information structure consists of signals about their own payoffs and distributions of these signals.
The precision of information is determined by the signal distributions.
Acquiring more precise signals is more costly, and each player considers the trade-off between the benefits of acquiring more precise information and the associated costs.
Thus, players strategically choose which signals to acquire and the distributions to maximize their expected payoffs net of the information costs.

Formally modeling costly information acquisition poses several challenges.
The domain of information structures is infinite-dimensional, which makes it difficult to characterize players' optimal choices and to derive equilibrium conditions. 
To address these issues, I adopt the framework of rational inattention \citep{matejka2015}. 
The theory of rational inattention provides a well-suited framework for analyzing the players' optimal strategies in games with costly information acquisition.
In particular, within this framework, players' optimal strategies yield a tractable characterization of Nash equilibrium \citep{Denti2023, yang2015, yang2020, montes2022}.
This tractability simplifies equilibrium analysis and leads to a feasible model estimation.

In equilibrium, the model captures heterogeneity in players' ability to acquire information through variation in the information strategies across players.
This contrasts with previous models that impose an exogenous information structure, where all players are assumed to have the same or arbitrarily specified information.
By allowing endogenous information acquisition, the model provides a unified framework to analyze how differences in information strategies, information costs, and player characteristics jointly determine equilibrium outcomes.

Next, I show that the model primitives are identified in both parametric and semiparametric settings.
The key assumption for the identification is the exclusion restriction.
That is, a player-specific variable affects only one player's payoff with no effect on the rival players' payoffs.
The exclusion restriction is not a strong assumption and is commonly used in the literature on discrete games \citep{Tamer2003, ciliberto2009, aradillas2010, grieco2014, magnolfi2022, koh2023}.
I illustrate both parametric and semiparametric identification results via Monte Carlo experiments.

Finally, as an empirical application, I use my framework to analyze the entry decisions of airlines in the U.S.
I specify a standard parametric linear payoff model and examine both the airlines' strategic behaviors and their costly information acquisition.
The results indicate that the airlines' baseline profits are negative, which creates an incentive for them to acquire information about their profits.
Furthermore, the results suggest that low-cost carriers tend to acquire less information than other airlines, which is consistent with their business models that focus on cost efficiency and simplified operations.
The analysis highlights how differences in firms' information strategies can explain observed heterogeneity in market entry behavior and competition.

This paper contributes to the literature on econometric discrete games \citep{BR1991, berry1992, Tamer2003, seim2006, sweeting2009, ciliberto2009, aradillas2010, grieco2014, magnolfi2022, koh2023, xie2024}.
Many of these studies assume that players have complete or incomplete information (or perfectly private information) among the players.
A common feature across this literature is that players' information is exogenously specified by the econometrician.
In contrast, this paper assumes that players have incomplete information about their own payoffs and does not impose any fixed information structure, instead allowing it to be optimally chosen by players. 

The discrete game model in this paper closely resembles recent work by \citet{grieco2014, magnolfi2022, koh2023}, which relaxes standard information assumptions and leverages the information structure to model econometric discrete games under flexible and weak information assumptions.
These studies adopt a novel approach to model player's information with the signals and the distribution of signals, allowing for a more flexible representation of players' information.
However, in these papers, the information structure is still given exogenously by the econometrician and remains fixed.
In contrast, this paper endogenizes the information structure by allowing players to optimally choose the information structure through the information strategy.
By incorporating the cost of information acquisition, this paper fills a gap in the recent literature and contributes to a more realistic framework for modeling behavior and strategies in discrete games.

This paper also builds on a literature on costly information acquisition and rational inattention. 
Most of the recent papers theoretically analyze decision problems of a single agent and determine the agent's optimal behavior or decision \citep{matejka2015, steiner2017,fosgerau2020}.
\citet{matejka2015} show that a rationally inattentive decision maker's optimal decision rule follows a logit-like formula, linking the rational inattention framework to discrete choice models.
\citet{fosgerau2020} generalizes a decision problem in \citet{matejka2015}, suggesting a generalized information cost function and proving that the rational inattention discrete choice problem and the additive random utility discrete choice model are equivalent under this information cost function.
\citet{steiner2017} extends \citet{matejka2015}'s approach and constructs a rationally inattentive decision maker's problem in a dynamic environment.
While I share major assumptions on the setup of the decision problem with these papers, my work extends the analysis to a multi-agent decision problem -- discrete games -- where strategic interaction plays a key role.

The related literature on game theory with costly information acquisition or rationally inattentive players includes \citet{yang2015, yang2020, Denti2023, montes2022}.
\citet{yang2015} investigates a two-player, two-action coordination investment game, and \citet{montes2022} focuses on attention-move games.
In these papers, players have imperfect information about the payoff-relevant state and optimally acquire information.
\citet{Denti2023} extends this framework to a game where players acquire not only about the state but also each other's information.
I adopt the structure of these theoretical models and expand them into an empirical model.
Thus, my model enables the quantitative analysis and estimation of players' information acquisition strategies using data. 
By bridging the gap between theory and empirical application, the model enhances the ability to validate theoretical predictions and better understand strategic behavior in environments characterized by costly information acquisition.

The remainder of this paper is organized as follows. 
Section \ref{sec:model} sets up an econometric model for discrete games with costly information acquisition.
Section \ref{sec:identification} provides the parametric and semiparametric identification.
In Section \ref{sec:montecarlo}, I present identification via Monte Carlo experiments and discuss the estimation method. 
Section \ref{sec:application} applies my framework to the U.S. airline industry to analyze the entry decisions. 
Finally, Section \ref{sec:conc} concludes.
All proofs are in Appendix.

\textbf{Notation.}
The boldface letters, e.g., $\mathbf{X}$ and $\mathbf{x}$, represent vectors.
The capital letters, e.g., $Y_{i}$, denote random variables, while the lowercase letters, e.g., $y_{i}$, represent their realized values.

\section{Model}
\label{sec:model}
In this section, I describe a general model of $2 \times 2$ game with costly information acquisition and provide a two-player entry game as a running example.

Let $\mathcal{I} = \{1, 2\}$ be the set of players, and let the players be indexed by $i=1, 2$.
By convention, I denote the rival of player $i$ as $-i$ or $j$.
\footnote{I consider a two-player game, and denoting the rival by $j$ does not violate the notation, as I can set $j = 3-i$.}
Each player $i$ simultaneously chooses an action $Y_{i}$ from their action set, denoted by $\mathcal{Y}_{i}$.
The action set contains two elements, namely $\mathcal{Y}_{i} = \{0,1\}$.
Let $\mathcal{Y} = \mathcal{Y}_{1} \times \mathcal{Y}_{2}$ represent the set of action profiles, and let $\mathbf{Y} = (Y_{1}, Y_{2})^{\top} \in \mathcal{Y}$ denote the action profile.

\subsection{Payoff Structure}
\label{subsec:payoff}
The payoff structure of player $i$ is additively separable and given by the following equation:
\begin{equation}\label{eq1:payoff}
    U_{i}(Y_{i}, Y_{j}, \varepsilon_{i}; \mathbf{X}, \mathbf{Z}_{i}) = 
    \begin{cases}
        u_{i}(Y_{j}; \mathbf{X}, \mathbf{Z}_{i}) + \varepsilon_{i} & \text{if } Y_{i} = 1, \\
        0, & \text{if } Y_{i} = 0,
    \end{cases}
\end{equation}
where $\mathbf{X} \in \mathcal{X} \subseteq \mathbb{R}^{d_{x}}$ is a vector of market-specific variables and $\mathbf{Z}_{i} \in \mathcal{Z}_{i} \subseteq \mathbb{R}^{d_{z}}$ is a vector of player-specific variables.
The payoff function for action $Y_{i}=1$, $u_{i}: \mathcal{Y}_{j} \times \mathcal{X} \times \mathcal{Z}_{i} \to \mathbb{R}$, can be specified parametrically or nonparametrically.
The payoff shock $\varepsilon_{i} \in \mathcal{E}_{i}$ is unobservable to player $i$ (and to the researcher) when making decisions, and I interpret it as the uncertainty in payoffs.
The payoff for action $Y_{i} = 0$ is normalized to zero.
In the next example below, I demonstrate parametric and nonparametric payoff functions.

The vector of market-specific variables $\mathbf{X}$ is common to all players and has influence on both players' payoffs.
On the contrary, the vector of player-specific variables $\mathbf{Z}_{i}$ only affects player $i$'s payoff, but not on rival's payoff.
As shown in the payoff function \eqref{eq1:payoff}, player $i$'s payoff depends on $\mathbf{X}$ and $\mathbf{Z}_{i}$, but not on $\mathbf{Z}_{j}$.
The variable $\mathbf{Z}_{i}$ is also referred to as the player-specific payoff shifter in the literature and satisfies the exclusion restriction assumption \citep{aradillas2010, bajari2010, xie2024}.
The exclusion restriction assumption plays a key role in the identification, which will be shown in the next section.
Examples of the market-specific variables $\mathbf{X}$ include market size, population, and consumer preference, whereas examples of player-specific variables $\mathbf{Z}_{i}$ include costs, technology, and product quality.
I assume that the variables $(\mathbf{X}', \mathbf{Z}_{i}', \mathbf{Z}_{j}')'$ are observable to all players and the researcher.
Furthermore, I do not restrict the support of each variables.
That is, the support of each variables of $\mathbf{X}$ and $\mathbf{Z}_{i}$ can be either discrete or continuous. 

Given the payoff structure \eqref{eq1:payoff}, the payoff uncertainty arises from two sources; the rival's action $Y_{j}$ and the unobservable payoff shock $\varepsilon_{i}$.
When making a decision, each player does not observe either the realized rival's action or the payoff shocks.
In particular, the unobservable payoff shock $\varepsilon_{i}$ directly affects the payoff, which creates an incentive for player $i$ to acquire information about it.
At the same time, because rival's decision is determined strategically, player $i$ needs to form a belief about rival's behavior.
Thus, the uncertainty in payoffs makes player's decision problem strategic since players weigh the value of acquiring information about their own payoff shocks while anticipating how rival will behave in equilibrium.

Since the main focus of this paper is on $2 \times 2$ games, the payoff function \eqref{eq1:payoff} can be expressed in a standard form commonly used in the literature.
Define two functions $\pi_{i}(\cdot,\cdot)$ and $\delta_{i}(\cdot,\cdot)$ as follows:
\begin{align*}
    \pi_{i}(\mathbf{X}, \mathbf{Z}_{i}) &= u_{i}(Y_{j} = 0; \mathbf{X}, \mathbf{Z}_{i}), \\
    \delta_{i}(\mathbf{X}, \mathbf{Z}_{i}) &= u_{i}(Y_{j} = 1; \mathbf{X}, \mathbf{Z}_{i}) - u_{i}(Y_{j} = 0; \mathbf{X}, \mathbf{Z}_{i}).
\end{align*}
I refer to the function $\pi_{i}(\cdot,\cdot)$ as the base payoff and the function $\delta_{i}(\cdot,\cdot)$ as the strategic effect.
Models without a strategic effect can be considered as a single-agent model where a player chooses action in isolation without strategic interaction.
Thus, I focus on models with non-zero strategic effect $\delta_i(\cdot,\cdot) \neq 0$ for all $\mathbf{X}$ and $\mathbf{Z}_{i}$ for each player $i$.

The sign of the strategic effect $\delta_{i}(\cdot,\cdot)$ depends on the strategic incentives of the game.
The strategic effect is positive if the players want to coordinate, and the example is timing of radio commercials \citep{sweeting2009}.
On the other hand, the strategic effect is negative if the players want to differentiate, and the example is entry games.
In a two-player entry game, the base payoff $\pi_{i}(\cdot,\cdot)$ corresponds to the monopoly profit, and the sum $\pi_{i}(\cdot,\cdot) + \delta_{i}(\cdot,\cdot)$ corresponds to the duopoly profit.

Using the above functions of the base payoff and the strategic effect, I can rewrite the payoff function as follows:
\begin{equation}\label{eq2:payoff}
    U_{i}(Y_{i}, Y_{j}, \varepsilon_{i}; \mathbf{X}, \mathbf{Z}_{i}) = 
    \begin{cases}
        \pi_{i}(\mathbf{X}, \mathbf{Z}_{i}) + \delta_{i}(\mathbf{X}, \mathbf{Z}_{i}) \cdot \mathds{1}(Y_{j}=1) + \varepsilon_{i}, & \text{if } Y_{i} = 1, \\
        0, & \text{if } Y_{i} = 0,
    \end{cases}
\end{equation}
where $\mathds{1}(\cdot)$ is the indicator function. 
Given the additive separable payoff shock $\varepsilon_i$, the original payoff function \eqref{eq1:payoff} can be transformed into a sum of the base payoff and the competitive effect without imposing any restrictions.
The analysis of the rest of this paper is mainly based on the payoff function \eqref{eq2:payoff}.

\begin{example}[A Two-player Entry Game] \label{example:entry}
    Consider a two-player entry game as in \citet{BR1991} and \citet{Tamer2003}.
    There are two players, $\mathcal{I} = \{1, 2\}$.
    The players are firms, and each firm simultaneously chooses either to enter the market $Y_{i} = 1$ or not to enter the market $Y_{i} = 0$.
    The profit function for firm $i$ is defined by 
    \begin{equation*}
        \Pi_{i}(Y_{i}, Y_{j}, \varepsilon_{i}; \mathbf{X}, \mathbf{Z}_{i}) = 
        \begin{cases}
        \pi_{i}(\mathbf{X}, \mathbf{Z}_{i}) + \delta_{i}(\mathbf{X}, \mathbf{Z}_{i}) \cdot \mathds{1}(Y_{j} = 1) + \varepsilon_{i}, & \text{if } Y_{i} = 1, \\
        0, & \text{if } Y_{i} = 0.
        \end{cases}
    \end{equation*}

    Firm $i$ earns zero profit if it chooses not to enter $Y_{i}= 0$.
    If firm $i$ chooses to enter, it gains monopoly profit if the rival stays out $Y_{j} = 0$ or duopoly profit if the rival enters $Y_{j} = 1$.
    In this model, $\delta_{i} (\cdot,\cdot)$ represents the competitive effect and is assumed to be negative.
    The matrix of the payoff structure is summarized as the table \ref{table1}.

    Payoffs $\pi_{i}(\cdot,\cdot)$ and $\delta_{i}(\cdot,\cdot)$ may be specified either as nonparametric functions or as parametric functions of covariates.
    The existing literature on entry games \citep{BR1991, Tamer2003, ciliberto2009} assumes linear parametric payoff functions, e.g.,
    \begin{align*}
        \pi_{i}(\mathbf{X}, \mathbf{Z}_{i}) &= \mathbf{X} \alpha_{i} + \mathbf{Z}_{i} \beta_{i}, \\
        \delta_{i}(\mathbf{X}, \mathbf{Z}_{i}) &= \delta_{i}.
    \end{align*}
    \hfill $\square$
    
    \begin{table}[htbp]
        \centering
        \begin{tabular}{lcc}
                                        & $Y_{2} = 0$ & $Y_{2} = 1$ \\ \cline{2-3}
        \multicolumn{1}{l|}{$Y_{1} = 0$} & \multicolumn{1}{c|}{$0, \;0$} & \multicolumn{1}{c|}{$0, \;\pi_{2}(\mathbf{X}, \mathbf{Z}_{2})+ \varepsilon_{2}$}\\ \cline{2-3}
        \multicolumn{1}{l|}{$Y_{1} = 1$} & \multicolumn{1}{c|}{$\pi_{1}(\mathbf{X}, \mathbf{Z}_{1}) + \varepsilon_{1}, \;0$} & \multicolumn{1}{c|}{$\pi_{1}(\mathbf{X}, \mathbf{Z}_{1}) + \delta_{1}(\mathbf{X}, \mathbf{Z}_{1}) + \varepsilon_{1}, \; \pi_{2}(\mathbf{X}, \mathbf{Z}_{2}) + \delta_{2}(\mathbf{X}, \mathbf{Z}_{2}) + \varepsilon_{2}$}\\ \cline{2-3}
        \end{tabular}
        \caption{The payoff matrix of $2 \times 2$ entry game in Example \ref{example:entry}.}
        \label{table1}
    \end{table}

\end{example}

\subsection{Information Structure}
Every player $i$ does not directly observe their own payoff shocks $\varepsilon_{i}$ and rival's payoff shock $\varepsilon_{j}$ when they make a decision.
The payoff shock $\varepsilon_{i}$ is realized and observable to each player $i$ after the realization of action profile $\mathbf{Y}$.
While the players are uncertain about the exact value of $\varepsilon_{i}$, they have a prior belief on the payoff shocks $\varepsilon = (\varepsilon_{i}, \varepsilon_{j})$.
The payoff shock $\varepsilon$ is distributed according to the distribution $F_{i}$.
I refer to the distribution $F_{i}$ as player $i$'s prior belief.
The following assumption summarizes the prior beliefs.

\begin{assumption}\label{assump:payoffshock}
    \begin{enumerate}
        \item The payoff shocks $\varepsilon_{i}$ and $\varepsilon_{j}$ are independent conditional on $(\mathbf{X}', \mathbf{Z}_{i}', \mathbf{Z}_{j}')'$.
        Moreover, for each player $i$, the payoff shock $\varepsilon_{i}$ is independent of $(\mathbf{X}', \mathbf{Z}_{i}', \mathbf{Z}_{j}')'$.

        \item The prior belief $F_{i}(\varepsilon)$ is absolutely continuous with respect to the Lebesgue measure. 
        That is, the prior $F_{i}(\varepsilon)$ is continuously differentiable and strictly increasing almost everywhere over the support of $\varepsilon$.
    \end{enumerate}
\end{assumption}

I model players' information through the signals they receive and the corresponding distributions, following \citet{magnolfi2022} and \citet{koh2023}.
Each player $i$ receives a private signal $\tau_{i}^{\mathbf{X}, \mathbf{Z}} \in \mathcal{T}_{i}^{\mathbf{X}, \mathbf{Z}}$, where $\mathcal{T}_{i}^{\mathbf{X}, \mathbf{Z}}$ denotes the set of all possible signals which player $i$ may receive.\footnote{Here, $(\mathbf{X}, \mathbf{Z})$ is a shorthand for $(\mathbf{X}, \mathbf{Z}_{i}, \mathbf{Z}_{j}).$}
The private signal $\tau_{i}^{\mathbf{X}, \mathbf{Z}}$ may contain information about the payoff shock $\varepsilon$, and its precision depends on the signal distribution.

Formally, the information structure $S_{i}^{\mathbf{X}, \mathbf{Z}}$ of player $i$ is defined as a pair consisting of the set of signals and the probability distributions over the signals:
    \begin{equation}
        S_{i}^{\mathbf{X}, \mathbf{Z}} = \left( \mathcal{T}_{i}^{\mathbf{X}, \mathbf{Z}}, \{P_{i}^{\mathbf{X}, \mathbf{Z}}(\cdot | \varepsilon): \varepsilon\in \mathcal{E}\} \right), 
    \end{equation}
where $\{P_{i}^{\mathbf{X}, \mathbf{Z}}(\cdot | \varepsilon)\}$ represents the family of conditional distributions of signals given the payoff shock $\varepsilon$. 
Each signal $\tau_{i}^{\mathbf{X}, \mathbf{Z}} \in \mathcal{T}_{i}^{\mathbf{X}, \mathbf{Z}}$ is realized according to its corresponding probability $P_{i}^{\mathbf{X}, \mathbf{Z}}( \tau_{i}^{\mathbf{X}, \mathbf{Z}} \, | \, \varepsilon)$.

One advantage of utilizing the information structure is that it allows us to express any type of information assumption from the existing literature.
To illustrate the concept of the information structure, consider the following two-player entry game example.

\begin{example}[Example \ref{example:entry} continued] 
    In \citet{BR1991}, \citet{Tamer2003}, and \citet{ciliberto2009}, firm $i$ not only observes his own payoff shock $\varepsilon_{i}$, but also rival's payoff shock $\varepsilon_{j}$.
    Thus, the firms have complete information.

    The complete information can be interpreted as follows: each player $i$ receives the exact value of $\varepsilon_{i}$ through the signal $\tau_{i}$.
    Formally, I can express the complete information structure as follows.
    The signal space for each firm is defined by $\mathcal{T}_{1} = \mathcal{E}_{1} \times \mathcal{E}_{2}$ and $\mathcal{T}_{2} = \mathcal{E}_{1} \times \mathcal{E}_{2}$.
    The distribution over signals for each firm is defined by
        \begin{align*}
            & P_{1} \left( \tau_{1} = (\varepsilon_{1}, \varepsilon_{2}) \, | \, \varepsilon_{1}, \varepsilon_{2} \right) = 1, \\
            & P_{2} \left( \tau_{2} = (\varepsilon_{1}, \varepsilon_{2}) \, | \, \varepsilon_{2}, \varepsilon_{1} \right) = 1.
        \end{align*} 
    \hfill $\square$
\end{example}

\begin{example}
    I describe another information structure from \citet{seim2006}, \citet{aradillas2010}, and \citet{bajari2010}.
    Firm $i$ only observes his own payoff shock $\varepsilon_{i}$, but does not receive any signals about their rival's payoff shock $\varepsilon_{j}$.
    In this situation, the firms have perfectly private information or incomplete information.
    The signal space for each firm is expressed as, for example,  $\mathcal{T}_{1} = \mathcal{E}_{1}\times \left\{ -\infty \right\}$ and $\mathcal{T}_{2} = \mathcal{E}_{2} \times \left\{ -\infty \right\}$.
    The distribution over signals for each firm is 
        \begin{align*}
            & P_{1} \left( \tau_{1} = (\varepsilon_{1}, -\infty)  \,|\, \varepsilon_{1}, \varepsilon_{2} \right) = 1, \\
            & P_{2} \left( \tau_{2} = (\varepsilon_{2}, -\infty)  \,|\, \varepsilon_{2}, \varepsilon_{1} \right) = 1.
        \end{align*}
    
    Alternatively, the signal spaces can be defined as $\mathcal{T}_{1} = \mathcal{E}_{1}$ and $\mathcal{T}_{2} = \mathcal{E}_{2}$.
    If the signal $\tau_{i}$ and rival's payoff shock $\varepsilon_{j}$ are conditionally independent given own payoff shock $\varepsilon_{i}$,
    \footnote{Once player $i$'s own payoff shock $\varepsilon_{i}$ is given, the signal $\tau_{i}$ does not provide any information about rival's payoff shock $\varepsilon_{j}$.} the distribution of the signals can be expressed as
        \begin{align*}
            & P_{1} \left( \tau_{1} = \varepsilon_{1}\,|\, \varepsilon_{1}, \varepsilon_{2} \right) = P_{1} \left( \tau_{1} = \varepsilon_{1}  \,|\, \varepsilon_{1} \right) = 1, \\
            & P_{2} \left( \tau_{2} = \varepsilon_{2} \,|\, \varepsilon_{2}, \varepsilon_{1} \right) = P_{2} \left( \tau_{2} = \varepsilon_{2} \,|\, \varepsilon_{2} \right) = 1.
        \end{align*}
    This information structure is identical to the one in \citet{seim2006}, \citet{aradillas2010}, and \citet{bajari2010}.
    \hfill $\square$
\end{example}

\begin{example}
    I can also consider a situation where each firm $i$ has no information about his own payoff shock $\varepsilon_{i}$.
    In this case, each firm's signal space is singleton, e.g., $\mathcal{T}_{i} = \left\{ -\infty \right\}$ for all realization of $\varepsilon$, and the distribution over the signal is, for $i=1,2$,
    \begin{equation*}
        P_{i} \left( \tau_{i} = -\infty \,|\, \varepsilon_{i}, \varepsilon_{j} \right) = 1.
    \end{equation*} 
    \hfill $\square$       
\end{example}

\subsection{Timeline of the Game}
The timing of the Bayesian game is as follows. 
In period 1, two players enter the game, and a vector of the covariates $(\mathbf{x}', \mathbf{z}_{1}', \mathbf{z}_{2}')'$ is realized.
In addition, the payoff shock $\varepsilon$ is drawn by nature.
Each player has limited information about their payoffs due to the unobserved payoff shock. 
In period 2, each player determines how much information to acquire and pays information costs.
I refer to this strategy as information strategy.
In period 3, the players receive private signals and update beliefs about the payoff shock.
In period 4, the players simultaneously choose actions $Y_{i}$ as a Bayesian expected utility maximizer.
I refer to this strategy as action strategy.
In period 5, the players' payoffs are realized.
This is summarized in the Figure \ref{fig:timing}.

\vskip 20pt

\begin{figure}[!t]
    \centering
    \begin{tikzpicture}[scale=1]
        \draw[black, thick, ->] (-6,0) -- (8,0);
        \filldraw[black] (-5,0) circle (2pt);
        \draw[black, ->] (-5,-2) -- (-5,-0.15);
        \node[below, align=center] at (-5,-2) {\small $(\mathbf{x}', \mathbf{z}_{1}', \mathbf{z}_{2}')'$ \\are realized};
        
        \filldraw[black] (-2,0) circle (2pt);
        \draw[black, ->] (-2,-2) -- (-2,-0.15);
        \node[below, align=center] at (-2,-2) {\small Select an \\information \\ strategy};
        \filldraw[black] (1,0) circle (2pt);
        \draw[black, ->] (1,-2) -- (1,-0.15);
        \node[below, text width=3.5cm, align=center] at (1,-2) {\small Receive a \\ signal and \\update beliefs};
        
        \filldraw[black] (4,0) circle (2pt);
        \draw[black, ->] (4,-2) -- (4,-0.15);
        \node[below, text width=3cm, align=center] at (4,-2) {\small Choose \\an action $Y_{i}$};
        
        \filldraw[black] (7,0) circle (2pt);
        \draw[black, ->] (7,-2) -- (7,-0.15);
        \node[below, text width=3.5cm, align=center] at (7,-2) {\small Payoffs $U_{i}$, \\ the shock $\varepsilon_{i}$ are observed};
    \end{tikzpicture}
    \caption{Timing of the Bayesian game}
    \label{fig:timing}
\end{figure}
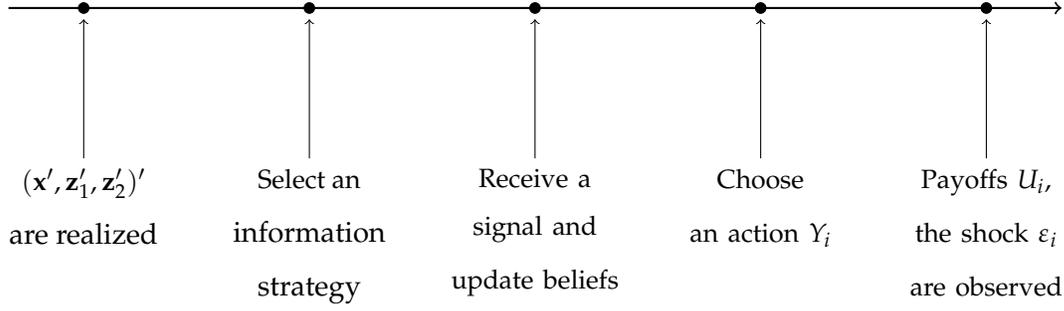

\subsection{Costly Information Acquisition}
Players have prior beliefs about the payoff shock $\varepsilon$ and choose how much information to acquire to form posterior beliefs about the payoffs.
However, acquiring information is costly.
According to the rational inattention theory \citep{matejka2015, caplin2019}, players cannot gather and process all available information, as attention is scarce and limited.
Moreover, acquiring information requires both time and effort, further adding to the cost of the information acquisition process.

Player $i$ cannot influence the realization of $\varepsilon$, but can freely select which signals to observe and their distributions.
The precision of these signals depends on the chosen distribution, and selecting more informative signals corresponds to acquiring more precise information.
However, there is a cost of acquiring information with more informative signals being more costly.
Given the prior belief $F_{i}$, observed covariates $(\mathbf{X}', \mathbf{Z}')'$, and the information cost function, player $i$ optimally chooses the information structure 
\begin{equation*}
    S_{i}^{\mathbf{X}, \mathbf{Z}} = \left(\mathcal{T}_{i}^{\mathbf{X},\mathbf{Z}}, \{P_{i}^{\mathbf{X},\mathbf{Z}}(\cdot|\varepsilon)\}\right),
\end{equation*}
which reflects a trade-off between the benefits of acquiring precise information and the costs of doing so.
As a result, players obtain partially informative signals, which lie between fully informative and completely uninformative signals.

In specifying information cost function, I adopt the entropy-based mutual information function widely used in the rational inattention literature \citep{matejka2015}.
First, I introduce the entropy function, which measures uncertainty or unpredictability of an event.
Next, I define the mutual information function using the entropy function.
The mutual information function measures the expected uncertainty reduction.
Finally, I define the information cost function, which is proportional to the mutual information.

Let $X$ be a continuous random variable with the density function $f(x)$.
Following \citep{cover2006}, the entropy of $X$ is defined as follows
\footnote{The random variable $X$ corresponds to the payoff shock $\varepsilon$ in the model. In the econometrics literature, the payoff shock follows the normal distribution or the type 1 extreme value distribution. On the other hand, theory literature assumes a discrete payoff shock in many cases. The entropy function with a discrete random variable $X$ is defined by $H(X) = H(p) = \sum_{x \in \mathcal{X}} p(x) \ln p(x)$ where $p(x)$ is a probability mass function.}: 
    \begin{equation*}
        H(f) = H(X) = -\int_{x\in\mathcal{X}} f(x) \ln f(x) \, \mathrm{d}x.
    \end{equation*}

To measure the acquired information, I use the mutual information.
Let $X$ and $Y$ be two random variables with the joint density function $f(x,y)$ and marginal density functions $f(x)$ and $f(y)$.
The mutual information $I(X,Y)$ between two random variables $X$ and $Y$ is defined as follows: 
    \begin{align*}
        I(X,Y) &= \int_{\mathcal{Y}} \int_{\mathcal{X}} f(x,y) \ln \left( \frac{f(x,y)}{f(x) f(y)} \right) \; \mathrm{d}x \, \mathrm{d}y \nonumber \\
        &=  H(X) - \mathrm{E}_{Y}\left[H(X \, | \, Y) \right].
    \end{align*}
The first term $H(X)$ denotes the entropy of $X$, which measures uncertainty of $X$ without any information about $Y$.
The second term $H(X \, | \, Y)$ denotes the entropy of $X$ given $Y$, which represents the uncertainty of $X$ once $Y$ is known.
Thus, the mutual information $I(X,Y)$ quantifies the expected uncertainty reduction of $X$ due to the knowledge of $Y$, which corresponds to the amount of information acquired.
In other words, the mutual information measures the decrease in prior uncertainty after information acquisition.

Finally, the information cost function $C(X,Y)$ is defined to be proportional to the mutual information, which measures the amount of information acquired.
Formally, the information cost function is expressed as:
    \begin{equation*}
        C(X,Y) = \lambda \cdot I(X,Y),
    \end{equation*} 
where $\lambda > 0$ is unit cost of information and $I(X,Y)$ is the mutual information between $X$ and $Y$.

Based on the above definitions, the information cost function in my setting is defined as follows.
First, I define the mutual information between the prior belief about the payoff shock and the posterior belief after observing the signal
    \begin{equation}\label{eq:mutualinfo}
        I(F_{i}, P_{i}^{\mathbf{X}, \mathbf{Z}}) = H(F_{i}(\cdot)) - \mathrm{E}_{\tau_{i}} \left[H(P_{i}^{\mathbf{X}, \mathbf{Z}}(\cdot|\tau_{i}^{\mathbf{X}, \mathbf{Z}})) \right],
    \end{equation}
where $H(F_{i}(\cdot))$ is the entropy of the prior belief $F_{i}(\varepsilon)$, and $H(P_{i}^{\mathbf{X}, \mathbf{Z}}(\cdot|\tau_{i}^{\mathbf{X}, \mathbf{Z}}))$ is the entropy of the posterior belief $P_{i}^{\mathbf{X}, \mathbf{Z}}(\varepsilon|\tau_{i}^{\mathbf{X}, \mathbf{Z}})$ after observing the signal $\tau_{i}^{\mathbf{X}, \mathbf{Z}}$.
This measures the information acquired about the payoff shock $\varepsilon$ by observing the signal $\tau_{i}^{\mathbf{X}, \mathbf{Z}}$.

Next, the information cost function is defined as:
\begin{equation}\label{eq:infocost}
        C_{i}(F_{i}, P_{i}^{\mathbf{X}, \mathbf{Z}}) = \lambda_{i} \cdot I(F_{i}, P_{i}^{\mathbf{X}, \mathbf{Z}}),
    \end{equation}
where the unit cost of information $\lambda_{i} > 0$.
The information cost function is proportional to the mutual information $I(F_{i}, P_{i}^{\mathbf{X}, \mathbf{Z}})$.
If a signal $\tau_{i}^{\mathbf{X}, \mathbf{Z}}$ does not contain any information about the payoff shock $\varepsilon$ or is independent of the payoff shock, the mutual information is zero, $I(F_{i}, P_{i}^{\mathbf{X}, \mathbf{Z}}) = 0$, and so is the information cost $C_{i}(F_{i}, P_{i}^{\mathbf{X}, \mathbf{Z}}) = 0$.
The following example illustrates the mutual information and the information cost function.

\begin{example}
    Consider a simple entry decision with unobservable payoff shocks given as $\varepsilon \in \{G, B\}$.
    Good market conditions are represented by $\varepsilon = G$, while bad market conditions are represented by $\varepsilon = B$.
    Suppose that the firm's prior belief about the payoff shock is given as $P(G) = 0.5$ and $P(B) = 0.5$.
    The entropy of the prior belief is 
    \begin{equation*}
        H(P) = -\sum_{x \in \{G, B\}} P(x) \ln P(x) \approx 0.693.
    \end{equation*}
    Now, consider two specific information structures $S_{1} = (\{0,1\}, P_{1})$ and $S_{2} = (\{0,1\}, P_{2})$.
    Each information structure generates signals $\tau \in \{0, 1\}$, where signal $\tau = 1$ can be interpreted as a recommendation to enter the market, and $\tau = 0$ as a recommendation not to enter.
    Let the signal distributions of the $S_{1}$ be
    \begin{equation*}
        P_{1}(\tau = 1 | G) = P_{1}(\tau = 0 | B) = 0.8, \quad P_{1}(\tau = 1 | B) = P_{1}(\tau = 0 | G) = 0.2,
    \end{equation*}
    and the signal distributions of the $S_{2}$ be
    \begin{equation*}
        P_{2}(\tau = 1 | G) = P_{2}(\tau = 0 | B) = 0.6, \quad P_{2}(\tau = 1 | B) = P_{2}(\tau = 0 | G) = 0.4.
    \end{equation*}
    Applying Bayes' rule, the posterior beliefs under $S_{1}$ are
    \begin{align*}
        &P_{1}(G|1) = 0.8, \quad P_{1}(B|1) = 0.2, \\
        &P_{1}(G|0) = 0.2, \quad P_{1}(B|0) = 0.8,
    \end{align*}
    and the mutual information is:
    \begin{equation*}
        I(P, P_{1}(\cdot|\tau)) = H(P) - \mathrm{E}_{\tau}[H(P_{1}(\cdot|\tau))] \approx 0.193.
    \end{equation*}
    Assuming the unit cost of information to be $\lambda = 1$, the information cost for $S_{1}$ is:
    \begin{equation*}
        C(P, P_{1}) \approx 0.193.
    \end{equation*}

    Similarly, the posterior beliefs under $S_{2}$ are:
    \begin{align*}
        &P_{2}(G|1) = 0.6, \quad P_{2}(B|1) = 0.4, \\
        &P_{2}(G|0) = 0.4, \quad P_{2}(B|0) = 0.6,
    \end{align*}
    with the mutual information:
    \begin{equation*}
        I(P, P_{2}(\cdot|\tau)) = H(P) - \mathrm{E}_{\tau}[H(P_{2}(\cdot|\tau))] \approx 0.020,
    \end{equation*}
    and the information cost:
    \begin{equation*}
        C(P, P_{2}) \approx 0.020.
    \end{equation*}

    Therefore, we conclude that the information structure $S_{1}$ is more informative than $S_{2}$. 
    This result arises because the signals in $S_{1}$ are more correlated with the true market conditions.
    Under $S_{1}$, the signal correctly matches the underlying market condition with higher probability than under $S_{2}$.
    As a result, the posterior beliefs under $S_{1}$ are more informative, leading to a larger expected reduction in uncertainty.
    \hfill $\square$
\end{example}

The intuition of the entropy-based information cost function can be explained as follows.
Consider a situation where a player asks a series of binary questions (yes or no questions) at a cost per question to determine the outcome of the payoff shocks.
The more questions he asks, the more information he can acquire, but also the higher the costs he pays.
This cost is proportional to the number of questions he asked and reduction in uncertainty about the payoff shocks.
Initially, the player faces a certain level of ex-ante uncertainty, measured by the entropy of the prior belief.
As he gathers information, this uncertainty is reduced.
The mutual information function precisely quantifies this average reduction in uncertainty, reflecting how much the player learns about the true payoff shocks by observing a signal.
Therefore, a higher mutual information value implies a greater reduction in uncertainty, which, in this model, corresponds to a higher information cost.

\subsection{Strategy and Equilibrium}
Each player's decision problem involves two strategies.
First, in period 2 of Figure \ref{fig:timing}, the player chooses the information strategy.
The information strategy is the selection of the optimal information structure $S_{i}^{\mathbf{X},\mathbf{Z}} = (\mathcal{T}_{i}^{\mathbf{X}, \mathbf{Z}}, P_i^{\mathbf{X},\mathbf{Z}}(\cdot \,|\,\varepsilon_{i}))$ given the information cost function \eqref{eq:infocost}.
\footnote{Player $i$'s payoff shock $\varepsilon_{i}$ is independent from $\varepsilon_{j}$ conditional on $(\mathbf{X}', \mathbf{Z}')'$ by Assumption~\ref{assump:payoffshock}. 
I further assume that the signal $\tau_{i}^{\mathbf{X},\mathbf{Z}}$ is independent from $\varepsilon_{j}$ given $\varepsilon_{i}$.
That is, $\varepsilon_{j}$ is the private information of rival, and the signal $\tau_{i}^{\mathbf{X},\mathbf{Z}}$ does not provide any information about $\varepsilon_{j}$ once $\varepsilon_{i}$ is given.
Thus, the signal distribution $P_{i}^{\mathbf{X},\mathbf{Z}}(\cdot|\varepsilon)$ can be simplified to $P_{i}^{\mathbf{X},\mathbf{Z}}(\cdot|\varepsilon_{i})$.}
In other words, it is the choice of how precise the signals should be, balancing the benefits of better information against the costs of acquiring it.

Second, once the player has chosen the optimal information structure, received the signals $\tau_{i}^{\mathbf{X},\mathbf{Z}}$, the player updates his prior beliefs to form the posterior beliefs using Bayes' rule: 
\begin{equation*}
    P_{i}^{\mathbf{X},\mathbf{Z}}(\varepsilon_{i}\,|\,\tau_{i}^{\mathbf{X},\mathbf{Z}}) = \frac{P_{i}^{\mathbf{X},\mathbf{Z}}(\tau_{i}^{\mathbf{X},\mathbf{Z}}\,|\,\varepsilon_{i}) f_{i}(\varepsilon_{i})}{P_{i}^{\mathbf{X},\mathbf{Z}}(\tau_{i}^{\mathbf{X},\mathbf{Z}})},
\end{equation*}
where $f_{i}(\cdot)$ is the density of the prior belief $F_{i}(\cdot)$.
After forming the posterior beliefs $P_{i}^{\mathbf{X},\mathbf{Z}}(\cdot\,|\,\tau_{i}^{\mathbf{X},\mathbf{Z}})$, the player chooses an action $Y_{i}$.
This stage corresponds to periods 3 and 4 in Figure \ref{fig:timing}.
The action strategy is a mapping from the signal to the action space $\mathcal{Y}_{i}$, denoted by $\sigma_{i}(\tau_{i}^{\mathbf{X},\mathbf{Z}})$.

Formally, the action strategy is defined as the best response to the rival's action strategy given the optimal information structure and the beliefs about the rival's action $b_{j}(Y_{j}|\mathbf{X},\mathbf{Z})$:
\begin{equation*}
    \sigma_{i}^{\mathbf{X}, \mathbf{Z}}(\tau_{i}^{\mathbf{X}, \mathbf{Z}}) = Y_{i} \iff \arg\max_{Y_{i} \in \mathcal{Y}_{i}} \mathrm{E}_{\varepsilon_{i} \sim P_{i}^{\mathbf{X},\mathbf{Z}}(\cdot|\tau_{i}^{\mathbf{X}, \mathbf{Z}})} \left[ \sum_{Y_{j}} U_{i}(Y_{i}, Y_{j}, \varepsilon_{i};\mathbf{X},\mathbf{Z}_{i})\cdot b_{j}(Y_{j}|\mathbf{X}, \mathbf{Z}_{j}) \right],
\end{equation*}
where the expectation is taken with respect to the posterior belief $P_{i}^{\mathbf{X},\mathbf{Z}}(\varepsilon\,|\,\tau_{i}^{\mathbf{X},\mathbf{Z}})$.
\footnote{The expectation taken with respect to the posterior belief $P_{i}^{\mathbf{X},\mathbf{Z}}(\varepsilon_{i}\,|\,\tau_{i}^{\mathbf{X},\mathbf{Z}})$ is:
\begin{equation*}
    \int_{\varepsilon_{i}} \left[ \sum_{Y_{j}} U_{i}(Y_{i}, Y_{j}, \varepsilon_{i};\mathbf{X},\mathbf{Z}_{i}) \cdot b_{j}(Y_{j}|\mathbf{X},\mathbf{Z}_{j})\right] P_{i}^{\mathbf{X},\mathbf{Z}}(\varepsilon_{i}\,|\,\tau_{i}^{\mathbf{X},\mathbf{Z}}) \mathrm{d}\varepsilon_{i}.
\end{equation*}}

One important result of decision-making with the costly information acquisition is that the information strategy and action strategy can be combined into a single strategy under optimal behavior \citep{matejka2015, yang2015, yang2020}.
Since information acquisition is costly, players will avoid collecting useless or redundant information, which determines the optimal signal space.
In particular, the signal space of each player's optimal information structure contains at most as many signals as the number of elements of the action set $\mathcal{Y}_{i}$.
Moreover, under optimal behavior, each signal leads to exactly one action.
Distinguishing distinct signals that induce the same action is inefficient, as it requires more costs but gives the same payoff.
It is more efficient for the players to collapse such distinct signals into one single signal.
The following lemma from \citet{matejka2015, yang2015, yang2020} formalizes this property.

\begin{lemma}\label{lemma1}
    Let $S_{i}^{\mathbf{X}, \mathbf{Z}} = (\mathcal{T}_{i}^{\mathbf{X}, \mathbf{Z}}, P_{i}^{\mathbf{X}, \mathbf{Z}})$ be the optimal information structure for player $i \in \mathcal{I}$.
    Then, player $i$ always plays pure strategies, and the signal space $\mathcal{T}_{i}^{\mathbf{X}, \mathbf{Z}}$ contains no more elements than those in $\mathcal{Y}_{i}$.
\end{lemma}
\begin{proof}
    In the Appendix
\end{proof}

Lemma~\ref{lemma1} implies that the optimal information structure can be represented as
\begin{equation*}
    S_{i}^{\mathbf{X}, \mathbf{Z}} = \left(\mathcal{Y}_{i}, \{P_{i}^{\mathbf{X}, \mathbf{Z}}(Y_{i}|\varepsilon_{i})\}\right).
\end{equation*}
Here, the signal space contains two signals, i.e., $\mathcal{T}_{i}^{\mathbf{X},\mathbf{Z}} \equiv \mathcal{Y}_{i} = \{0,1\}$, with the corresponding signal distribution $P_{i}^{\mathbf{X}, \mathbf{Z}}(Y_{i}|\varepsilon_{i})$.
Intuitively, receiving the signal $\tau_{i}^{\mathbf{X}, \mathbf{Z}}=1$ can be interpreted as ``accept/yes", whereas $\tau_{i}^{\mathbf{X}, \mathbf{Z}}=0$ corresponds to ``reject/no".
The player's optimal action strategy is simply to follow the received signal.
Specifically, if player $i$ receives a signal of ``accept/yes", $\tau_{i}^{\mathbf{X}, \mathbf{Z}} = 1$, with probability $P_{i}^{\mathbf{X}, \mathbf{Z}}(1|\varepsilon_{i})$, then the player chooses action $Y_{i}=1$ as recommended.
Conversely, if player $i$ receives a signal of ``reject/no", $\tau_{i}^{\mathbf{X}, \mathbf{Z}} = 0$, with probability $P_{i}^{\mathbf{X}, \mathbf{Z}}(0|\varepsilon_{i})$, then the player chooses action $Y_{i}=0$ as recommended.
Thus, the action strategy is
\begin{equation*}
    \sigma_{i}^{\mathbf{X}, \mathbf{Z}}(\tau_{i}^{\mathbf{X},\mathbf{Z}}=Y_{i}) = Y_{i}.
\end{equation*}
Hence, once the optimal information structure is chosen, the action strategy is automatically determined by following the signal.
In this sense, it is without loss of generality to regard player $i$'s strategy as focusing solely on the information strategy.

Each player $i$ maximizes the ex-ante expected payoff less the information cost.
Let $b_{j}(Y_{j}|\mathbf{X}, \mathbf{Z})$ be player $i$'s belief about the rival's action $Y_{j}$.
Using the result of Lemma~\ref{lemma1}, player $i$'s ex-ante expected payoff maximization problem can be expressed as:
\begin{equation}\label{eq:optim}
    \max_{P_{i}^{\mathbf{X}, \mathbf{Z}}:\mathcal{E}_{i}\to(0,1)} \int_{\varepsilon_{i}} \sum_{Y_{i}} \sum_{Y_{j}} U_{i}(Y_{i}, Y_{j}, \varepsilon_{i}; \mathbf{X}, \mathbf{Z}_{i}) \cdot b_{j}(Y_{j}|\mathbf{X}, \mathbf{Z}) \cdot P_{i}^{\mathbf{X},\mathbf{Z}}(Y_{i}|\varepsilon_{i}) \, \mathrm{d} F_{i}(\varepsilon_{i}) - C_{i}(F_{i}, P_{i}^{\mathbf{X}, \mathbf{Z}})
\end{equation}
    subject to
\begin{equation*}
    C_{i}(F_{i}, P_{i}^{\mathbf{X}, \mathbf{Z}}) = \lambda_{i} I(F_{i}, P_{i}^{\mathbf{X}, \mathbf{Z}}), \quad 
    \sum_{Y_{i}} P_{i}^{\mathbf{X}, \mathbf{Z}} = 1, \quad
    P_{i}^{\mathbf{X}, \mathbf{Z}}(Y_{i}|\varepsilon) \geq 0, \quad \text{for all } i \in \mathcal{I},\,\varepsilon_{i} \in \mathcal{E}_{i}.
\end{equation*}

The solution to the optimization problem has been well established in \citet{matejka2015} and \citet{caplin2019}.
\footnote{\citet{matejka2015} established the necessary conditions for the rational inattention optimization problem, but did not prove sufficiency. 
Subsequently, \citet{caplin2019} derived the necessary and sufficient conditions using the concept of consideration sets.} 
Proposition~\ref{prop:optimsol} below characterizes the player's optimal strategy.

\begin{proposition}\label{prop:optimsol}
    Let $P_{i}^{\mathbf{X},\mathbf{Z}}(Y_{i}|\varepsilon_{i})$ be a strategy of player $i$ in the decision problem \eqref{eq:optim}.
    Then the strategy $P_{i}^{\mathbf{X},\mathbf{Z}}(Y_{i}|\varepsilon_{i})$ is optimal if and only if it satisfies the following:
    \begin{equation}\label{eq:optimsol}
        P_{i}^{\mathbf{X},\mathbf{Z}}(Y_{i}|\varepsilon_{i}) 
        = \frac{P_{i}^{\mathbf{X},\mathbf{Z}}(Y_{i}) \exp\left( \sum_{Y_{j}} U_{i}(Y_{i}, Y_{j}, \varepsilon_{i}; \mathbf{X}, \mathbf{Z}_{i}) b_{j}(Y_{j}|\mathbf{X}, \mathbf{Z}) \right)^{1/\lambda_{i}}}{\sum_{Y_{i'}}P_{i}^{\mathbf{X},\mathbf{Z}}(Y_{i}') \exp\left( \sum_{Y_{j}} U_{i}(Y_{i}', Y_{j}, \varepsilon_{i}; \mathbf{X}, \mathbf{Z}_{i}) b_{j}(Y_{j}|\mathbf{X}, \mathbf{Z})\right)^{1/\lambda_{i}}},\text{ a.s.}
    \end{equation}
    with
    $P_{i}^{\mathbf{X}, \mathbf{Z}}(Y_{i}) = \int_{\varepsilon_{i}} P_{i}^{\mathbf{X},\mathbf{Z}}(Y_{i}|\varepsilon_{i}) \,\mathrm{d}F_{i}(\varepsilon_{i})>0$.
\end{proposition}

\begin{proof}
    In the Appendix
\end{proof}

Next, I formally define the equilibrium of the game in the context of costly information acquisition.
The solution concept is the Bayesian-Nash equilibrium (BNE), which is the standard in the literature on incomplete information games.

\begin{definition}\label{def:equilibrium}
A Bayesian-Nash equilibrium (BNE) is a profile of strategies $(P_{i}^{\mathbf{X}, \mathbf{Z}}(Y_{i}))_{i \in \mathcal{I}}$ such that for each player $i \in \mathcal{I}$ with $j \neq i$,
\begin{enumerate}
    \item Given beliefs about the rival's action $b_{j}(Y_{j}|\mathbf{X}, \mathbf{Z})$, the strategy $P_{i}^{\mathbf{X}, \mathbf{Z}}$ solves the maximization problem \eqref{eq:optim}.
    \item Beliefs are consistent.
    That is, $b_{j}(Y_{j}|\mathbf{X}, \mathbf{Z})= P_{j}^{\mathbf{X}, \mathbf{Z}}(Y_{j})$.
\end{enumerate}
\end{definition}

The equilibrium concept in Definition~\ref{def:equilibrium} captures the strategic interactions between players in a game with costly information acquisition.
First, each player's strategy in a Bayesian-Nash equilibrium optimally balances the expected payoff from each action with the costs of acquiring information, conditional on the beliefs about rival's actions.
Second, consistency of beliefs ensures that each player's expectation about the rival's strategies aligns with the actual equilibrium strategies.
Finally, the equilibrium defined here generalizes the standard BNE concept to account for endogenous information acquisition, providing a foundation for empirical estimation of discrete games with costly information.

In the next subsection, I focus on the equilibrium in the choice probability space, which facilitates the characterization of players' strategies and the computation of equilibrium outcomes.
\subsection{Equilibrium Existence and Uniqueness}
Given the equilibrium strategies defined above, it is often convenient to represent the equilibrium in terms of conditional choice probabilities (CCPs), which aggregate over the unobserved payoff shocks. 
Associated with each player's equilibrium strategy, the CCP $P_{i}(Y_{i}|\mathbf{X}, \mathbf{Z}_{i}, \mathbf{Z}_{j})$ is defined as
\begin{equation}\label{eq:ccp-def}
    P_{i}(Y_{i}|\mathbf{X}, \mathbf{Z}_{i}, \mathbf{Z}_{j}) \equiv P_{i}^{\mathbf{X},\mathbf{Z}}(Y_{i}) 
    = \int_{\varepsilon_{i}} P_{i}^{\mathbf{X},\mathbf{Z}}(Y_{i}|\varepsilon_{i}) \, \mathrm{d}F_{i}(\varepsilon_{i}).
\end{equation}
Since the action is binary, the CCPs can be simplified to a single probability of choosing action $Y_{i}=1$: 
\begin{equation*}
    P_{i}(\mathbf{X}, \mathbf{Z}_{i}, \mathbf{Z}_{j}) \equiv P_{i}(Y_{i}=1|\mathbf{X}, \mathbf{Z}_{i}, \mathbf{Z}_{j}).
\end{equation*}

Moreover, the equilibrium condition can be expressed as a fixed-point equation.
Let  $\mathbf{P}(\mathbf{X}, \mathbf{Z}) = (P_{1}(\mathbf{X},\mathbf{Z}), P_{2}(\mathbf{X},\mathbf{Z}))'$ denote the vector of conditional choice probabilities.
For $i, j \in \mathcal{I}$ with $j \neq i$, define 
\begin{equation}\label{eq:psi-def}
    \psi_{i}(p_{i}, p_{j},\mathbf{X}, \mathbf{Z}) = \int_{\varepsilon_{i}} \frac{p_{i} \exp\left( \pi_{i}(\mathbf{X},\mathbf{Z}_{i}) + \delta_{i}(\mathbf{X},\mathbf{Z}_{i}) \cdot p_{j} + \varepsilon_{i} \right)^{1/\lambda_{i}}}{p_{i} \exp\left( \pi_{i}(\mathbf{X},\mathbf{Z}_{i}) + \delta_{i}(\mathbf{X},\mathbf{Z}_{i}) \cdot p_{j} + \varepsilon_{i} \right)^{1/\lambda_{i}} + 1 - p_{i}} \, \mathrm{d}F_{i}(\varepsilon_{i}).
\end{equation}
Stacking the two functions, define the mapping $\Psi:[0,1]^{2} \to [0,1]^{2}$ as:
\begin{equation}
    \Psi(p_{1}, p_{2}, \mathbf{X}, \mathbf{Z}) =
    \begin{pmatrix}
        \psi_{1}(p_{1}, p_{2}, \mathbf{X}, \mathbf{Z}) \\
        \psi_{2}(p_{2}, p_{1}, \mathbf{X}, \mathbf{Z})
    \end{pmatrix}.
\end{equation}
From equations~\eqref{eq:optimsol} and \eqref{eq:ccp-def}, the equilibrium choice probabilities satisfy 
\begin{equation*}
    P_{i}(\mathbf{X},\mathbf{Z}) = \psi_{i}(P_{i}(\mathbf{X}, \mathbf{Z}), P_{j}(\mathbf{X},\mathbf{Z}), \mathbf{X}, \mathbf{Z}).
\end{equation*}
Thus, the equilibrium conditional choice probabilities are characterized as the fixed-point solution of
\begin{equation}\label{eq:fixedpoint}
    \mathbf{P}(\mathbf{X},\mathbf{Z}) = 
    \Psi(\mathbf{P}(\mathbf{X}, \mathbf{Z}), \mathbf{X}, \mathbf{Z}).
\end{equation}
The following proposition establishes existence of a Bayesian Nash equilibrium.
The proof follows from Brouwer's fixed-point theorem.

\begin{proposition}
    Under the Assumption 1, there exists a Bayesian Nash equilibrium of the games with costly information acquisition.
\end{proposition}
\begin{proof}
    In the Appendix
\end{proof}

As shown in Figure \ref{fig:unique-equil}, the multiple solutions to equation~\eqref{eq:fixedpoint} may arise, a common feature in discrete choice games.
\citet{aradillas2010} investigates the conditions under which the equilibrium is unique.
Building on Gale-Nikaido conditions \citep{gale1965}, \citet{aradillas2010} derives the sufficient conditions for equilibrium uniqueness.
The following proposition formalizes uniqueness conditions in the context of the game with costly information acquisition.

\begin{proposition}\label{prop:unique}
    Let $\nabla_{\mathbf{P}}\Psi(\mathbf{P}(\mathbf{X},\mathbf{Z}),\mathbf{X},\mathbf{Z})$ denote the Jacobian of \eqref{eq:psi-def} evaluated at the equilibrium choice probabilities $\mathbf{P}(\mathbf{X},\mathbf{Z})$.
    The solution to the equation 
    \begin{equation*}
        \mathbf{P}(\mathbf{X}, \mathbf{Z}) - \Psi(\mathbf{P}(\mathbf{X}, \mathbf{Z}), \mathbf{X}, \mathbf{Z}) = 0
    \end{equation*}
    is unique if no principal minors of the matrix
    \begin{equation}\label{eq:jacobian}
        I_{2\times 2} - \nabla_{\mathbf{P}} \Psi(\mathbf{P}(\mathbf{X}, \mathbf{Z}), \mathbf{X}, \mathbf{Z})
    \end{equation}
    are zero, where $I_{2 \times 2}$ is the 2-dimensional identity matrix.
    That is,
    \begin{enumerate}
        \item the $(1,1)$ element of the matrix~\eqref{eq:jacobian} is positive, i.e., $1 - \dfrac{\partial\psi_{1}}{\partial P_{1}}\bigg|_{\mathbf{P}(\mathbf{X},\mathbf{Z})} > 0$,
        \item the $(2,2)$ element of the matrix~\eqref{eq:jacobian} is positive, i.e., $1 - \dfrac{\partial\psi_{2}}{\partial P_{2}}\bigg|_{\mathbf{P}(\mathbf{X},\mathbf{Z})} > 0$, and
        \item the determinant of the matrix~\eqref{eq:jacobian} is positive, i.e., $\det\!\big(I_{2\times2} - \nabla_{\mathbf{P}}\Psi(\mathbf{P}(\mathbf{X},\mathbf{Z}),\mathbf{X},\mathbf{Z})\big) > 0$.
    \end{enumerate}
\end{proposition}

The third condition in Proposition~\ref{prop:unique} is automatically satisfied if $\delta_{1}(\mathbf{X}, \mathbf{Z}_{1}) \times \delta_{2}(\mathbf{X}, \mathbf{Z}_{2}) < 0$, or if one player's action is strategic substitute and the other's is strategic complement. 
The panel (c) in Figure~\ref{fig:unique-equil} illustrates this case. 
The sign of $\delta_{i}(\mathbf{X}, \mathbf{Z}_{i})$ determines the slope of player $i$'s best response function.
A positive $\delta_{i}(\mathbf{X}, \mathbf{Z}_{i})$ indicates that the best response function is increasing, while a negative $\delta_{i}(\mathbf{X}, \mathbf{Z}_{i})$ indicates that it is decreasing.
Thus, when the two slopes have opposite signs, the best response functions are monotone in opposite directions and intersect at one point.

\begin{figure}[t!]
    \centering
    \caption{Unique Equilibrium and Multiple Equilibria}
    \begin{subfigure}[b]{0.495\textwidth}
        \centering
        \includegraphics[width=\linewidth]{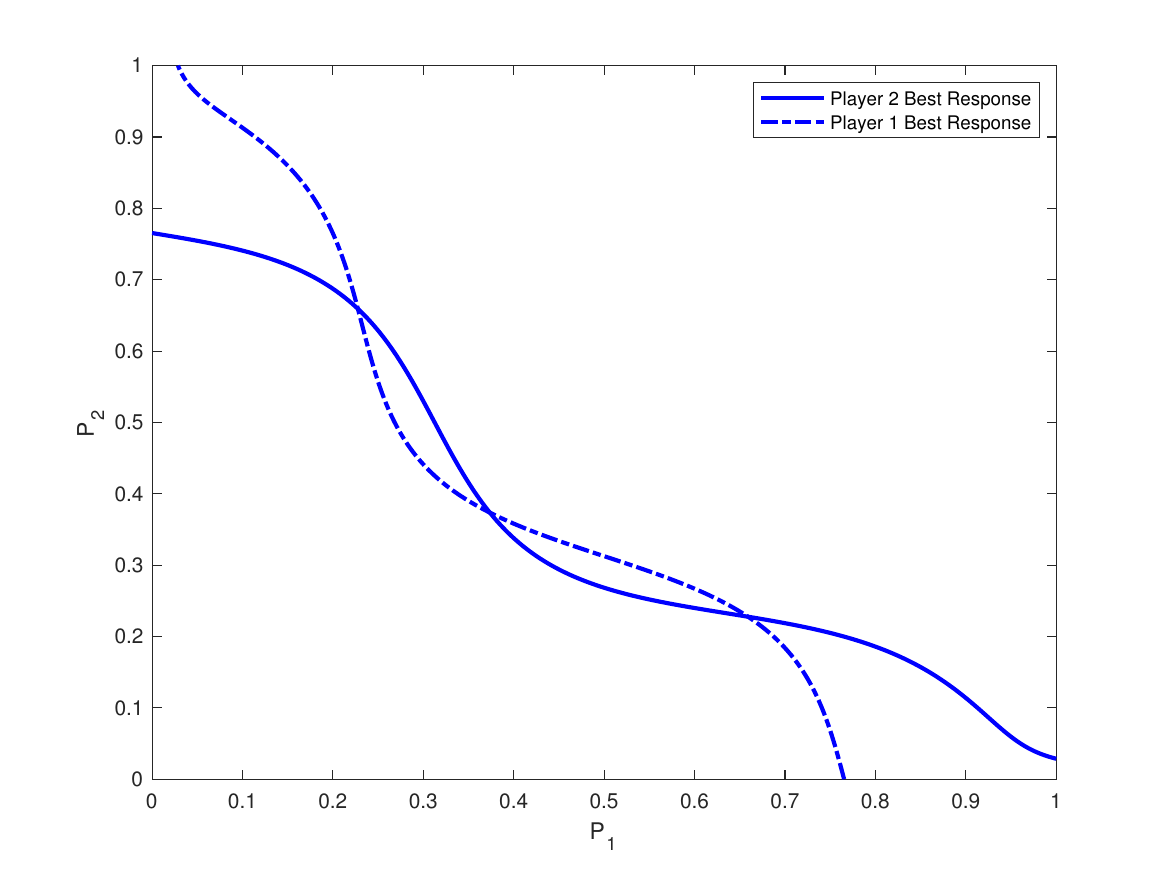}
        \caption{{\small Multiple Equilibria, $\delta_{1}\cdot\delta_{2}>0$}}    
    \end{subfigure}
    \hfill
    \begin{subfigure}[b]{0.495\textwidth}
        \centering
        \includegraphics[width=\linewidth]{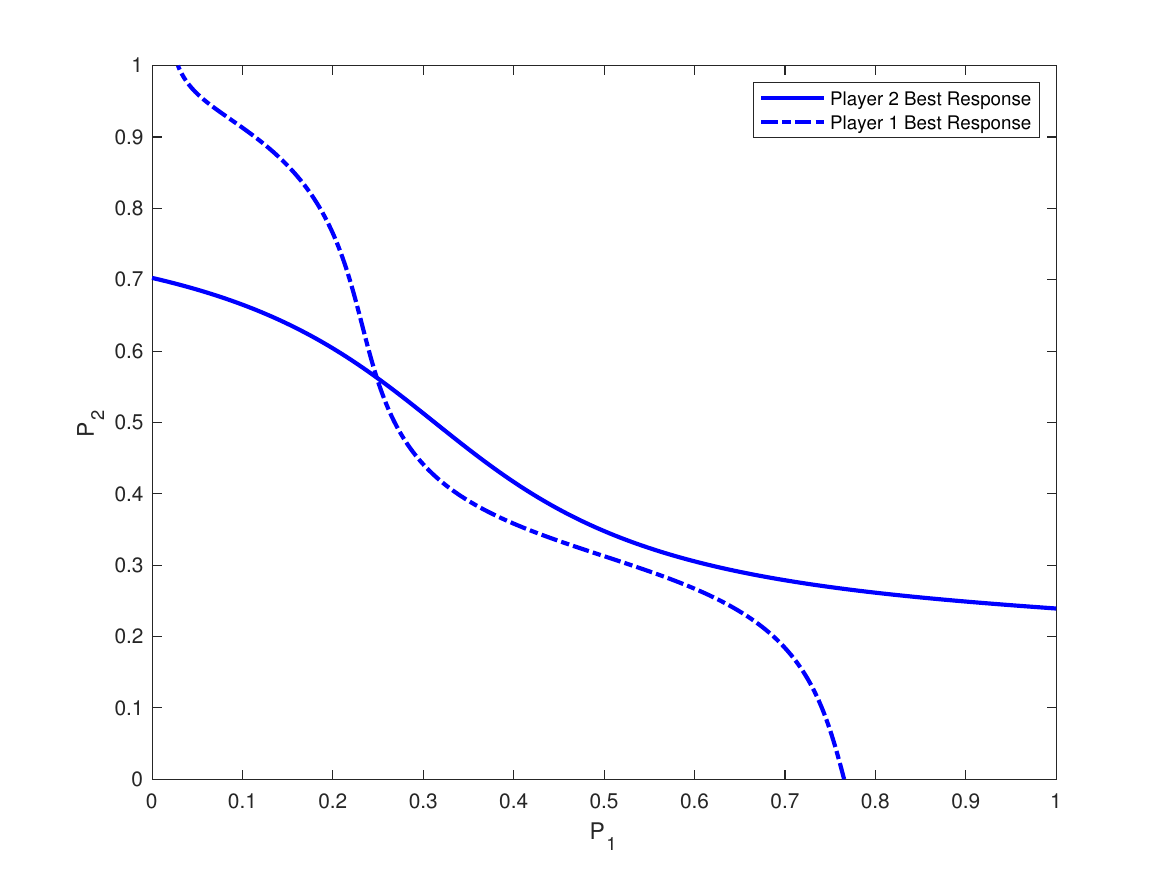}
        \caption{{\small Unique Equilibrium, $\delta_{1}\cdot\delta_{2}>0$}}
    \end{subfigure}
    \begin{subfigure}[b]{0.495\textwidth}
        \centering
        \includegraphics[width=\linewidth]{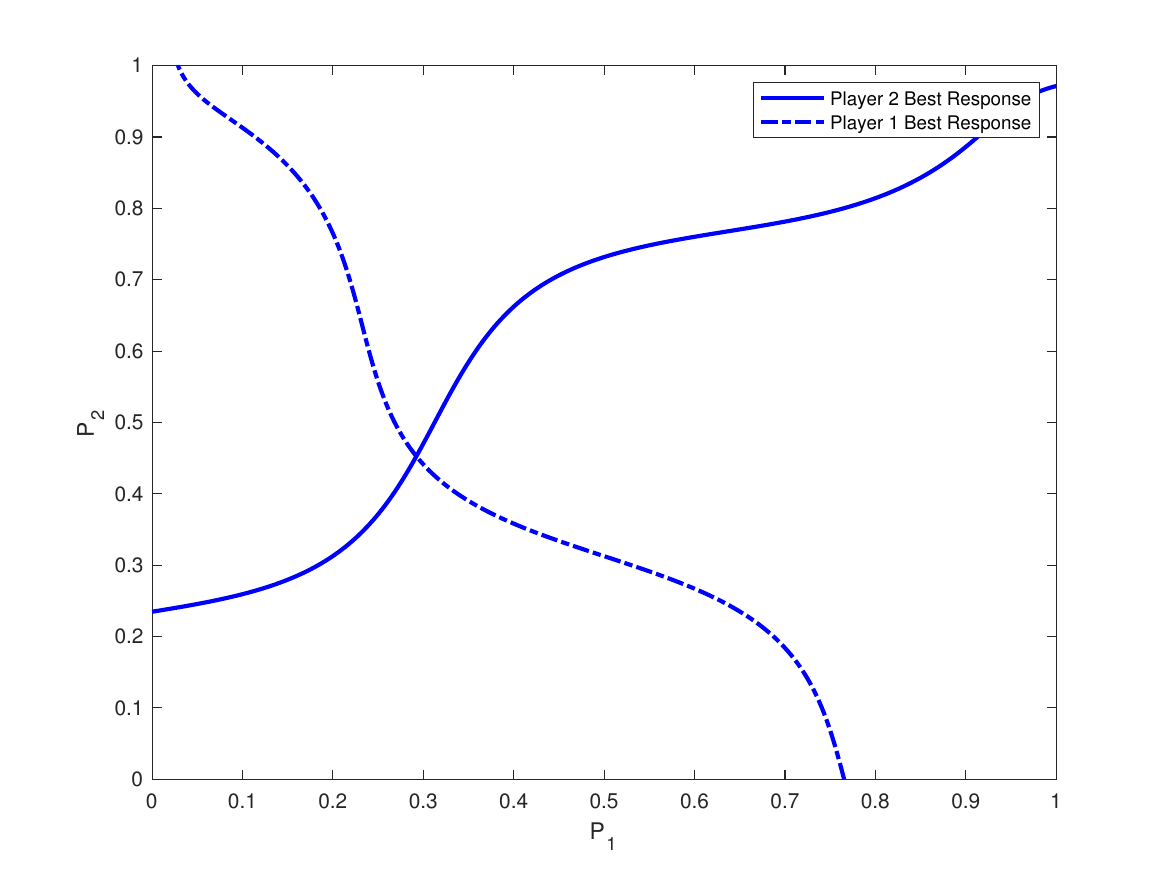}
        \caption{{\small Unique Equilibrium, $\delta_{1}\cdot\delta_{2}<0$}}
    \end{subfigure}
    \begin{minipage}
        {0.9\textwidth}
        \footnotesize
        \textit{Notes:} The figures illustrate the best response functions.
        The horizontal axis represents player 1's choice probability $P_{1}(\mathbf{X},\mathbf{Z})$, and the vertical axis represents player 2's choice probability $P_{2}(\mathbf{X},\mathbf{Z})$.
        The dotted line represents player 1's best response function, and the solid line represents player 2's best response function.
        The intersection points of the two lines indicate the equilibria of the game.
    \end{minipage}
    \label{fig:unique-equil}
\end{figure}

\section{Identification}
\label{sec:identification}
In this section, I describe the identification result of model primitives. 
First, I present the data generating process and the identification conditions. 
Next, I demonstrate that the model primitives, the payoff function, can be identified.

\subsection{Data Generating Process and Identification Conditions}
Suppose that an econometrician observes an independent dataset of $M$ different markets, indexed by $m = 1,2,\ldots, M$.
The dataset $\{ (\mathbf{y}_{m}', \mathbf{x}_{m}', \mathbf{z}_{1m}', \mathbf{z}_{2m}')'\}_{m=1}^{M} $ consists of the realized action profiles $\mathbf{y}_{m}$, vectors of market-specific variables $\mathbf{x}_{m}$, and a vector of player-specific variables $\mathbf{z}_{im}$ for each market $m=1,2,\ldots,M$.
I assume that both the market-specific and player-specific variables are observable to the econometrician.

\begin{assumption}\label{assump:dgp}
    The realizations of the vector $\{(\mathbf{y}_{m}', \mathbf{x}_{m}', \mathbf{z}_{1m}', \mathbf{z}_{2m}')'\}_{m=1}^{M}$ are independent and identically distributed across markets and observable to the econometrician.
\end{assumption}

Let $Z_{ik}$ denote the $k$-th variable in $\mathbf{Z}_{i}$, and 
let $\mathbf{Z}_{i,-k} = (Z_{i1}, \ldots, Z_{i,k-1}, Z_{i,k+1}, \dots, Z_{id_{z}})'$ denote the vector of $\mathbf{Z}_{i}$ excluding $Z_{ik}$.
I assume that the variable $Z_{ik}$ varies over its support $\mathcal{Z}_{ik}$, conditional on player $i$'s other variables $\mathbf{Z}_{i,-k}$, the rival's variables $\mathbf{Z}_{j}$, and the market-specific variables $\mathbf{X}$.

The supports of each $\mathbf{Z}_{i} \in \mathcal{Z}_{i} \subseteq \mathbb{R}^{d_{z}}$ can be either discrete or continuous.
If $\mathbf{Z}_{i}$ has a continuous support, then at least one of player-specific variables, $Z_{ik}$ for some $k$, is continuously distributed for $d_{z} \geq 1$.
In other words, the support of $\mathcal{Z}_{ik}$ contains a non-degenerate interval.

I further assume that the vector of player-specific variables $\mathbf{Z}_{i}$ enters only player $i$'s payoff function and is excluded from the rival's payoff function.
Changes in $\mathbf{Z}_{i}$ affect only player $i$'s payoff and do not directly affect the rival's payoff.
This assumption plays a key role in identification and is referred to as exclusion restriction.
Assumption \ref{assump:exclusion} summarizes the conditions on the player-specific variable $\mathbf{Z}_{i}$.

\begin{assumption}\label{assump:exclusion}
    \begin{enumerate}
        \item The player-specific variable $\mathbf{Z}_{i}$ can be discrete or continuous, and $\mathbf{Z}_{i}$ enters player $i$'s payoff but is excluded from the payoffs of other player.
        \item The variable $Z_{ik}$ varies over its support $\mathcal{Z}_{ik}$, conditional on $(\mathbf{X}', \mathbf{Z}_{i,-k}', \mathbf{Z}_{j}')'$.
    \end{enumerate}
\end{assumption}

Next, I assume that the equilibrium conditional choice probabilities (CCPs) can be consistently estimated or identified from the data.
Specifically, each player's conditional choice probabilities $P_{i}(\mathbf{X}, \mathbf{Z}_{i}, \mathbf{Z}_{j})$ are assumed to be known to the econometrician.
As defined in equation \eqref{eq:ccp-def}, the CCPs vary with both player-specific variables $(\mathbf{Z}_{i}', \mathbf{Z}_{j}')'$.

\begin{assumption}\label{assump:ccp}
    The conditional choice probabilities $P_{i}(\mathbf{X}, \mathbf{Z}_{i}, \mathbf{Z}_{j})$ vary with $(\mathbf{X}', \mathbf{Z}_{i}', \mathbf{Z}_{j}')'$. 
    Furthermore, the CCPs are identified or consistently estimated from the data.
\end{assumption}

Finally, I derive the relationship between equilibrium CCPs and the model primitives.
From the equations~\eqref{eq:psi-def} through \eqref{eq:fixedpoint}, the equilibrium CCPs satisfy
    \begin{equation}\label{eq:fixedpoint-full}
        P_{i}(\mathbf{X}, \mathbf{Z}) = \int_{\varepsilon_{i}} \frac{P_{i}(\mathbf{X}, \mathbf{Z}) \exp\left( \pi_{i}(\mathbf{X},\mathbf{Z}_{i}) + \delta_{i}(\mathbf{X}, \mathbf{Z}_{i})P_{j}(\mathbf{X},\mathbf{Z}) + \varepsilon_{i} \right)^{1/\lambda_{i}}}{P_{i}(\mathbf{X}, \mathbf{Z}) \exp\left( \pi_{i}(\mathbf{X},\mathbf{Z}_{i}) + \delta_{i}(\mathbf{X}, \mathbf{Z}_{i})P_{j}(\mathbf{X},\mathbf{Z}) + \varepsilon_{i} \right)^{1/\lambda_{i}} + (1 - P_{i}(\mathbf{X}, \mathbf{Z})) } \,\mathrm{d}F_{i}(\varepsilon_{i}),
    \end{equation}
for $i,j \in \mathcal{I}$ with $j \neq i$.
The terms in the exponent represent player $i$'s expected payoff, and the deterministic component of the expected payoff is
\begin{equation}\label{eq:deterministicpayoff}
    EU_{i}(\mathbf{X}, \mathbf{Z}) = \pi_{i}(\mathbf{X}, \mathbf{Z}_{i}) + \delta_{i}(\mathbf{X}, \mathbf{Z}_{i}) P_{j}(\mathbf{X}, \mathbf{Z}).
\end{equation}
Since the integrand of \eqref{eq:fixedpoint-full} is strictly increasing in $EU_{i}(\mathbf{X}, \mathbf{Z})$ and the prior belief $F_{i}(\varepsilon_{i})$ is bijective, there exists a strictly increasing function $G_{i}(\cdot)$ such that 
\begin{equation}\label{eq:ccp-inverse}
    P_{i}(\mathbf{X}, \mathbf{Z}) = G_{i} \left(\frac{1}{\lambda_{i}}EU_{i}(\mathbf{X}, \mathbf{Z}) \right) = G_{i}\left( \frac{1}{\lambda_{i}} \left( \pi_{i}(\mathbf{X}, \mathbf{Z}_{i}) + \delta_{i}(\mathbf{X}, \mathbf{Z}_{i}) P_{j}(\mathbf{X}, \mathbf{Z}) \right) \right).
\end{equation}
This equation shows that the equilibrium CCPs can be expressed as a function of the players' deterministic components of expected payoffs, providing a key link between the CCPs and the underlying structure of the model.

\subsection{Identification Result: Semiparametric Model}
Given these assumptions, I now turn to the identification of model primitives.
In this subsection, the model primitives --- the base payoffs $\pi_{i}(\cdot)$ and the competitive effects $\delta_{i}(\cdot)$ --- are nonparametric functions, and I demonstrate that they are identified.

Fix an arbitrary realization of market-specific variables and player $i$'s variables, $\mathbf{X} = \mathbf{x}$ and $\mathbf{Z}_{i} = \mathbf{z}_{i}$.
Given the equilibrium choice probabilities $P_{i}(\mathbf{x}, \mathbf{z}_{i}, \mathbf{Z}_{j})$, I can apply the mapping in equation~\eqref{eq:ccp-inverse} to recover $EU_{i}(\mathbf{x}, \mathbf{z}_{i}, \mathbf{Z}_{j})$.
Recall that $G_{i}(\cdot)$ is a bijective function and hence invertible.
Therefore,
\begin{equation}\label{eq:semiparam-id}
    G_{i}^{-1}(P_{i}(\mathbf{x}, \mathbf{z}_{i}, \mathbf{Z}_{j})) = \frac{1}{\lambda_{i}} \left( \pi_{i}(\mathbf{x}, \mathbf{z}_{i}) + \delta_{i}(\mathbf{x}, \mathbf{z}_{i}) P_{j}(\mathbf{x}, \mathbf{Z}_{j}, \mathbf{z}_{i}) \right).
\end{equation}

Given that $\lambda_{i}>0$, model primitives $\{\pi_{i}(\cdot),\delta_{i}(\cdot), \lambda_{i}\}$ are only identified up to scale.
Multiplying $\pi_{i}(\cdot)$, $\delta_{i}(\cdot)$ and $\lambda_{i}$ by a positive constant $c>0$ does not change the right-hand side of equation \eqref{eq:semiparam-id}.
I therefore normalize the unit cost of information $\lambda_{i}$ to one.
Under this normalization, the model primitives, $\pi_{i}(\cdot)$, and $\delta_{i}(\cdot)$ can be identified.

A necessary condition for identification of $\pi_{i}(\cdot)$ and $\delta_{i}(\cdot)$ is that, for each $\mathbf{x}$ and $\mathbf{z}_{i}$, there exist at least two distinct points, $\mathbf{z}_{j}^{1}$ and $\mathbf{z}_{j}^{2}$ in the support of $\mathbf{Z}_{j}$ conditional on $\mathbf{x}$ and $\mathbf{z}_{i}$.
A sufficient condition for identification is that the two distinct points satisfy the following:
\begin{equation*}
    P_{j}(\mathbf{x}, \mathbf{z}_{j}^{1}, \mathbf{z}_{i}) \neq P_{j}(\mathbf{x}, \mathbf{z}_{j}^{2}, \mathbf{z}_{i}).
\end{equation*}
By Assumption ~\ref{assump:ccp}, the conditional choice probabilities $P_{i}(\mathbf{x}, \mathbf{z}_{i}, \mathbf{Z}_{j})$ and $P_{j}(\mathbf{x}, \mathbf{Z}_{j}, \mathbf{z}_{i})$, vary sufficiently with $\mathbf{Z}_{j}$, so the inequality condition holds.
Therefore, variation in $P_{i}(\mathbf{x}, \mathbf{z}_{i}, \mathbf{Z}_{j})$ and $P_{j}(\mathbf{x}, \mathbf{Z}_{j}, \mathbf{z}_{i})$ across $\mathbf{Z}_j$ allows the identification of $\pi_{i}(\cdot)$ and $\delta_{i}(\cdot)$.

\begin{proposition}
    Suppose that Assumptions 1 through 4 hold and the unit cost of information $\lambda_{i}$ is normalized to 1.
    Suppose further that, for each $\mathbf{x} \in \mathcal{X}$ and $\mathbf{z}_{i} \in \mathcal{Z}_{i}$, there exist at least two distinct points, $\mathbf{z}_{j}^{1}$ and $\mathbf{z}_{j}^{2}$ in the support of $\mathbf{Z}_{j}$ conditional on $\mathbf{x}$ and $\mathbf{z}_{i}$, such that 
    \begin{equation*}
        P_{j}(\mathbf{x}, \mathbf{z}_{j}^{1}, \mathbf{z}_{i}) \neq P_{j}(\mathbf{x}, \mathbf{z}_{j}^{2}, \mathbf{z}_{i}).
    \end{equation*}
    Then, for each $i \in \mathcal{I}$, the base payoff $\pi_{i}(\mathbf{X}, \mathbf{Z}_{i})$ and the competitive effects $\delta_{i}(\mathbf{X}, \mathbf{Z}_{i})$ are nonparametrically identified.
\end{proposition}
\begin{proof}
    In the Appendix
\end{proof}

The proposition shows that, under the maintained assumptions and sufficient variation in $\mathbf{Z}_{j}$, which induces variation in $P_{i}(\mathbf{x}, \mathbf{z}_{i}, \mathbf{Z}_{j})$ and $P_{j}(\mathbf{x}, \mathbf{Z}_{j}, \mathbf{z}_{i})$, both the structural base payoffs and the strategic effects can be nonparametrically identified from the equilibrium choice probabilities. 
Intuitively, variation in $\mathbf{Z}_{j}$ changes both players' equilibrium choice probabilities without directly affecting the payoff functions of player $i$. 
By observing how player $i$'s choice probabilities respond to changes in $\mathbf{Z}_{j}$, the payoff primitives can be identified even in the presence of costly information acquisition.

\subsection{Identification Result: Parametric Model}
In this subsection, the model primitives are specified as parametric functions, and I show that the parameter vector is point identified.
Let the payoff function $\pi_{i}(\cdot)$ and the competitive effect $\delta_{i}(\cdot)$ be defined by parametric functions as follows:
\begin{align*}
    \pi_{i}(\mathbf{X}, \mathbf{Z}_{i}) &= \mathbf{X}\beta_{i} + \mathbf{Z}_{i}\gamma_{i}, \\
    \delta_{i}(\mathbf{X}, \mathbf{Z}_{i}) &= \delta_{i},
\end{align*}
where $\beta_{i} \in \mathbb{R}^{d_{x}}$, $\gamma_{i} \in \mathbb{R}^{d_{z}}$, and $\delta_{i} \in \mathbb{R}$ are constant parameters.
The competitive effect $\delta_{i}$ may also take a linear functional form, such as $\delta_{i}(\mathbf{X}, \mathbf{Z}_{i}) = \mathbf{X}\alpha_{i} + \mathbf{Z}_{i}\kappa_{i}$, where $\alpha_{i} \in \mathbb{R}^{d_{x}}$ and $\kappa_{i} \in \mathbb{R}^{d_{z}}$ are constant parameters.

The deterministic expected payoff \eqref{eq:deterministicpayoff} can be expressed as
\begin{equation*}
    EU_{i}(\mathbf{X}, \mathbf{Z}, \theta_{i}) = \mathbf{X}\beta_{i} + \mathbf{Z}_{i}\gamma_{i} + \delta_{i} \cdot P_{j}(\mathbf{X}, \mathbf{Z}, \theta_{j}),
\end{equation*}
where $\theta_{i} = (\beta_{i}', \gamma_{i}', \delta_{i})'$ is the parameter vector.

To identify the parameters, I adopt the strategy known as identification at infinity following \citet{Tamer2003}.
Identification requires an assumption that there exists one variable $Z_{ik}$ such that $\gamma_{ik} \neq 0$ and has an everywhere positive Lebesgue density conditional on $\mathbf{Z}_{i,-k}$.
This assumption ensures that the variation in $Z_{ik}$ can be used to identify the effect of $Z_{ik}$ on the payoff function, which is crucial for point identification of the parameter vector.
Moreover, this assumption is standard in the literature on discrete choice games \citep{Tamer2003,ciliberto2009}.
Under the assumption, the parameter vector $\theta = (\theta_{i}',\theta_{j}')'$ is point identified.

\begin{proposition}
Suppose that Assumptions 1 through 4 hold.
I further assume that there exists one variable $Z_{ik}$ such that $\gamma_{ik} \neq 0$ and has an everywhere positive Lebesgue density conditional on $\mathbf{Z}_{i,-k}$.
Then the parameter vector $\theta$ is point identified if the matrix $(\mathbf{X}, \mathbf{Z}_{i})$ has full column rank for all $i \in \mathcal{I}$.
\end{proposition}
\begin{proof}
    In the Appendix
\end{proof}

This result formalizes the point identification of the parametric version of the model.
While the nonparametric setting established identification of the underlying payoff primitives, this proposition demonstrates that, under rank conditions, the finite-dimensional parameter vector $\theta$ can be uniquely identified from the equilibrium choice probabilities.
Hence, the model provides a fully identified structure that can be directly estimated using standard econometric estimation methods.

\section{Monte Carlo Experiments}
\label{sec:montecarlo}
In this section, I illustrate the identification results in the Monte Carlo experiments.

\subsection{Parametric Model}
I conduct a Monte Carlo experiment to illustrate the parametric identification results of games with costly information acquisition.
I construct a sample of $M$ independent markets with two players and two actions.
The payoff function for each player is linear in parameters, and the payoff for each player from $Y_{i} = 1$ is given by 
\begin{align*}
    U_{1}(Y_{2}, \varepsilon_{1}; Z_{1}, \theta) &= 1.8 + 0.5Z_{1} -1.3 \cdot \mathbb{1}(Y_{2} = 1) + \varepsilon_{1}, \\
    U_{2}(Y_{1}, \varepsilon_{2}; Z_{2}, \theta) &= 1.6 + 0.8Z_{2} -1.3 \cdot \mathbb{1}(Y_{1} = 1) + \varepsilon_{2},
\end{align*}
where $\theta$ is the vector of parameters.
The payoff from $Y_{i} = 0$ is zero for both players.
The true parameters in the data generating process $\theta^{0} = (\beta_{1}^{0}, \gamma_{1}^{0}, \beta_{2}^{0}, \gamma_{2}^{0},\delta_{1}^{0}, \delta_{2}^{0})$ are given by $(1.8, 0.5, 1.6, 0.8, -1.3, -1.3)$.
I assume no market-specific variables in this exercise, and the player-specific variable $Z_{i}$ is randomly drawn from the uniform distribution with support $[0,1]$.
Each player's prior belief $F_{i}(\varepsilon_{i})$ follows the normal distribution with mean zero and variance $4$.

The sample data are generated from the unique equilibrium in each market conditional on $(Z_{1},Z_{2})$, and I estimate the model using the maximum likelihood (ML) with the constrained optimization approach \citep{su2014}.
One challenge in estimating discrete games with incomplete information is the potential existence of multiple equilibria, which makes computing all the equilibria often infeasible.
As a result, the ML estimation is often demanding and computationally burdensome.
In this exercise, however, the equilibrium is unique by construction, and thus the ML estimation is feasible.
Among the alternative estimation methods, such as two-step estimators and pseudo-likelihood estimators, the constrained ML approach is known to be efficient \citep{su2014}.

Let $\mathbf{P}_{m}=(P_{1m}, P_{2m})'$ be a vector of choice probabilities in market $m$.
Given the parameter vector $\theta$ and the data $\mathbf{z}_{m} = (z_{1m}, z_{2m})'$ and $\mathbf{y}_{m} = (y_{1m}, y_{2m})'$, the log likelihood function in market $m$ is 
\begin{equation*}
    \ell_{m}(\mathbf{y}_{m}, \mathbf{z}_{m}, \mathbf{P}_{m}, \theta) = \sum_{i=1,2} y_{im} \log \left( P_{im} \right) + (1 - y_{im}) \log \left( 1 - P_{im} \right).
\end{equation*}
The log likelihood for all markets $m=1,2,\ldots,M$ is 
\begin{equation}\label{eq:likelihood}
    \mathcal{L}(\mathbf{P}, \theta) = \sum_{m=1}^{M} \ell_{m}(\mathbf{y}_{m}, \mathbf{z}_{m}, \mathbf{P}_{m}, \theta),
\end{equation}
where $\mathbf{P} = (\mathbf{P}_{m})_{m=1}^{M}$ is the vector of choice probabilities for each market.

The probability $\mathbf{P}$ in the equation \eqref{eq:likelihood} is not necessarily an equilibrium choice probability. 
I impose the Nash equilibrium restriction \eqref{eq:fixedpoint} as a constraint.
Hence, the ML estimation is to solve the maximization problem:
\begin{equation}
    \max_{\mathbf{P}, \theta} \mathcal{L}(\mathbf{P}, \theta) \quad \text{s.t.} \quad \mathbf{P} - \Psi(\mathbf{P}, \mathbf{z}, \theta)= 0.
\end{equation}

I report the estimation results in Table \ref{tab-sim-param}.
The table shows the result for samples of $M=1,000$ markets and $R=300$ repetitions.
The result displays strong evidence of convergence of each estimates to the true parameters.

\begin{table}[ht]
    \centering
    \begin{tabular}{cccccccc}
        \hline\hline 
        Parameter & True Value & Mean & Median & S.D. & RMSE & Mean Bias & Median Bias \\ \hline
        $\beta_{1}$ & 1.8 & 1.7803 & 1.7432 & 0.4633 & 0.4638 & -0.0197 & -0.0568 \\ 
        $\gamma_{1}$ & 0.5 & 0.5050 & 0.5050 & 0.1558 & 0.1559 & 0.0050 & 0.0050 \\ 
        $\beta_{2}$ & 1.6 & 1.6172 & 1.6127 & 0.3679 & 0.3683 & 0.0172 & 0.0127 \\ \
        $\gamma_{2}$ & 0.8 & 0.8060 & 0.8092 & 0.1535 & 0.1537 & 0.0060 & 0.0092 \\ 
        $\delta_{1}$ & -1.3 & -1.2754 & -1.2194 & 0.5847 & 0.5853 & 0.0246 & 0.0806 \\ 
        $\delta_{2}$ & -1.3 & -1.3241 & -1.3207 & 0.4469 & 0.4476 & -0.0241 & -0.0207 \\ \hline
    \end{tabular}
    \caption{Monte Carlo Results: Parametric Model, $M=1000$}
    \label{tab-sim-param}
\end{table}

\subsection{Semiparametric Model}
Next, I illustrate the semiparametric identification results of games with costly information acquisition.
In this Monte Carlo experiment, the maximum likelihood estimation by the method of sieve \citep{chen2007} in implemented to estimate nonparametric functions of the base return and the competitive effect.

The data generating process mainly follows \citet{xie2024}.
I construct a sample of $M$ independent markets.
I consider two identical players in each market.
The payoff function for each player from $Y_{i} = 1$ is given by 
\begin{align*}
    \pi_{i}(Z_{i}) &= 3 - \log(1 + 2Z_{i}), \\
    \delta_{i}(Z_{i}) &= -\frac{4 \exp(Z_{i})}{1 + \exp(Z_{i})},
\end{align*}
where $Z_{i}$ is randomly drawn from the uniform distribution with the support $[0,3]$.
The payoff from $Y_{i} = 0$ is zero for both players.
Each player's prior belief $F_{i}(\varepsilon_{i})$ follows from the uniform distribution with the support $[-5,5]$.

Let $\theta_{i} = (\pi_{i}, \delta_{i})$ be model primitives.
$\theta_{i}$ contains all unknown functions which is to be estimated by the sieve MLE.
Moreover, each element of the model primitives $\theta_{i}$ satisfies $\pi_{i} \in \Pi$ and $\delta_{i} \in \Delta$, where $\Pi$ and $\Delta$ are the spaces of functions $\pi_{i}$ and $\delta_{i}$.
The approximating spaces, i.e. sieves, are represented by $\Pi_{M} \subseteq \Pi$ and $\Delta_{M} \subseteq \Delta$, and the sieve space for the model primitive $\theta_{i}$ is $\Theta_{M} = \Pi_{M} \times \Delta_{M}$.

I use the space of polynomials as the sieve space for both $\pi_{i}$ and $\delta_{i}$.
The sieve spaces are 
\begin{equation*}
    \Pi_{M} = \left\{ \sum_{l=0}^{d_{\pi, M}} \beta_{l} \cdot Z_{i}^{l} \,\middle|\, \beta_{l} \in \mathbb{R} \right\}, \quad \Delta_{M} = \left\{ \sum_{l=0}^{d_{\delta, M}} \beta_{l} \cdot Z_{i}^{l} \,\middle|\, \beta_{l} \in \mathbb{R}\right\}.
\end{equation*}

The equilibrium choice probabilities follow from \eqref{eq:ccp-inverse}
\begin{equation}
    P_{i}(Z_{i}, Z_{j}; \theta_{i}) = G_{i}\left( \pi_{i}(Z_{i}) + \delta_{i}(Z_{i}) P_{j}(Z_{j}, Z_{i};\theta_{j}) \right),
\end{equation}
where $G_{i}(\cdot)$ is the inverse of the uniform distribution with the support $[-5,5]$.

Under the regularity conditions \citep{chen2007}, the sieve MLE is consistent, and the model primitives $\theta_{i}$ can be estimated by maximizing the following log-likelihood function
\begin{equation}
    \max_{\theta_{i} \in \Theta_{M}} \sum_{i=1}^{2}\sum_{m=1}^{M} y_{im} \log \left( P_{im}(z_{i}, z_{3-i,m}, \theta_{i}) \right) + (1 - y_{im}) \log \left( 1 - P_{im}(z_{im}, z_{3-i,m}, \theta_{i}) \right) 
\end{equation}

Figure \ref{fig:simulation-siemle} shows the estimation results.
I use the sample of $M=400$ independent markets, and repeated $R=100$ times.
The degree of polynomial was set to 2.
The blue solid lines in each panel indicate the true value of $\pi_{i}$ and $\delta_{i}$.
The red solid lines in each panel represent the averages of the estimates from 100 repetition.
The red dashed lines show 90\% confidence intervals.
The sieve MLE give estimates that are consistent and close to the true model primitives.

\begin{figure}
        \centering
        \begin{subfigure}[t]{0.485\textwidth}
            \centering
            \includegraphics[width=\textwidth]{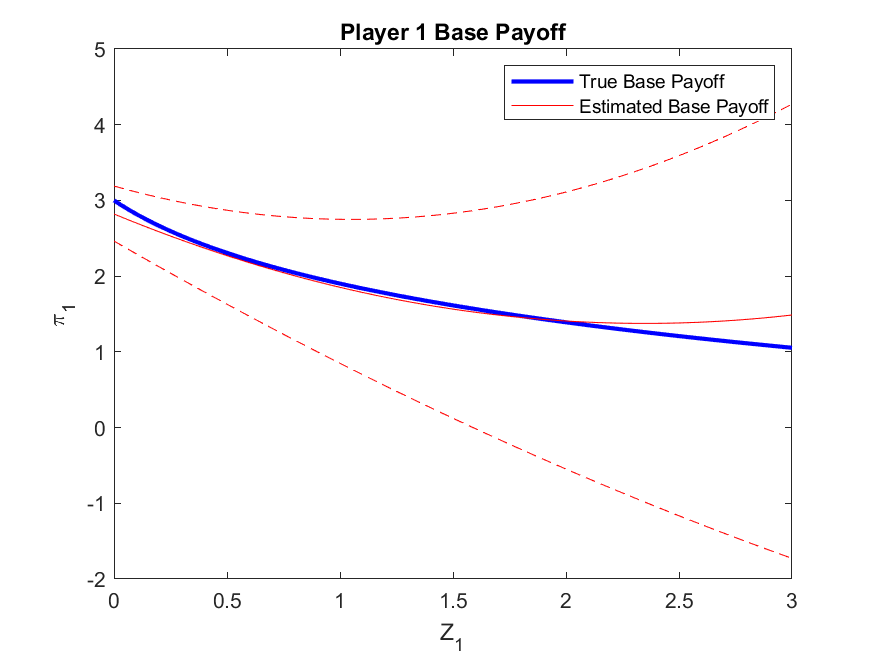}
            \caption{{\small Player 1's base return}}    
        \end{subfigure}
        \hfill
        \begin{subfigure}[t]{0.485\textwidth}
            \centering
            \includegraphics[width=\textwidth]{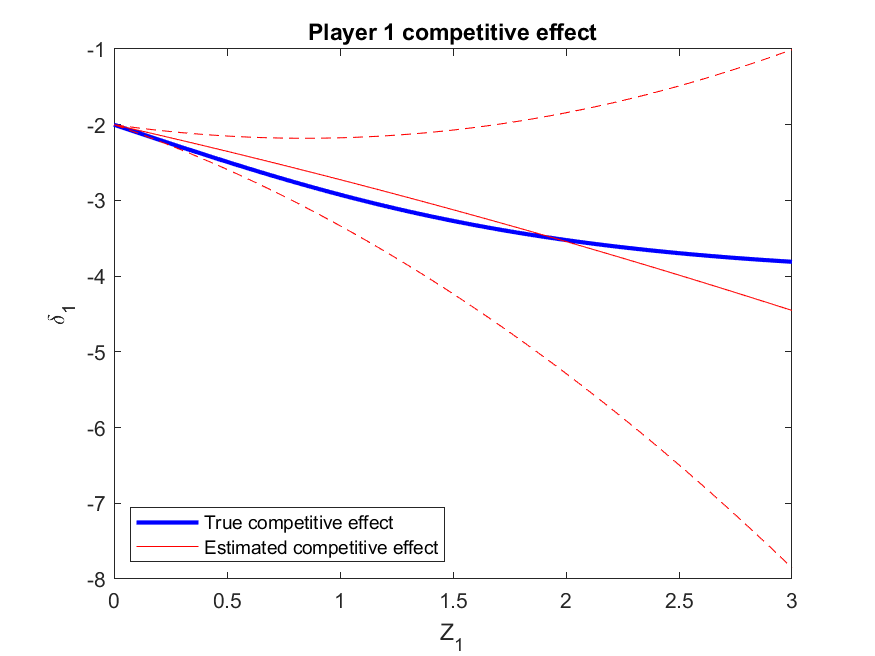}
            \caption{{\small Player 1's competitive effect}}
        \end{subfigure}
        \vskip\baselineskip
        \begin{subfigure}[t]{0.485\textwidth}
            \centering
            \includegraphics[width=\textwidth]{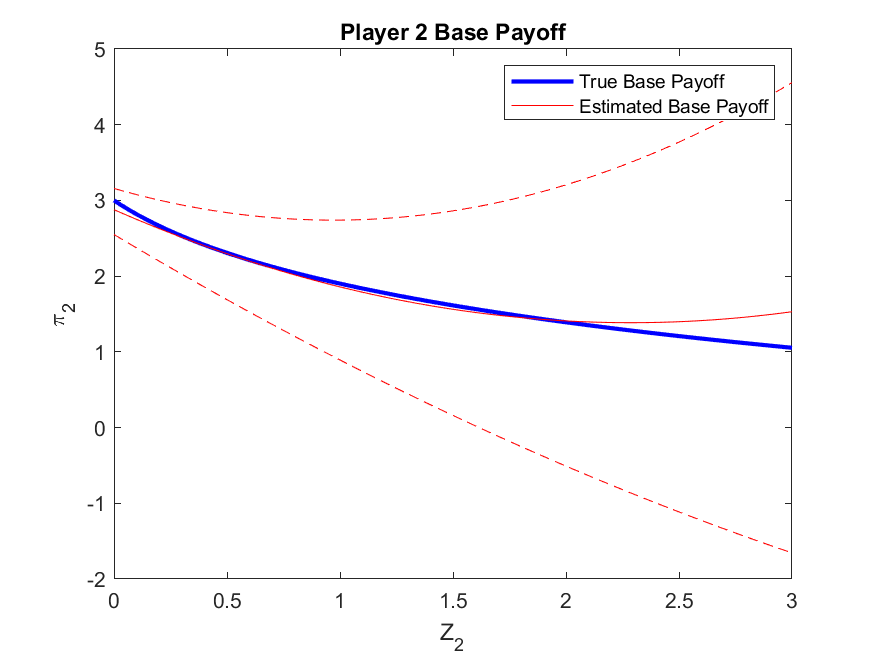}
            \caption{{\small Player 2's base return}}    
        \end{subfigure}
        \hfill
        \begin{subfigure}[t]{0.485\textwidth}
            \centering
            \includegraphics[width=\textwidth]{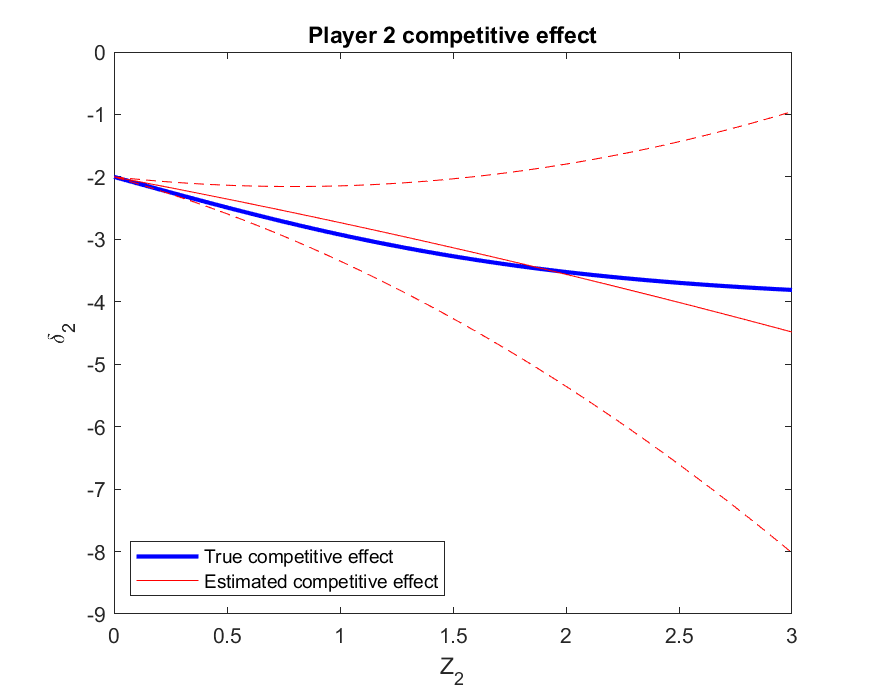}
            \caption{{\small Player 2's competitive effect}}
        \end{subfigure}
        \caption{Monte Carlo Results: Semiparametric Model, $M=400$.}
        \label{fig:simulation-siemle}
\end{figure}

\section{Empirical Application}
\label{sec:application}
In this section, I apply my model to the market-level entry decisions of airline firms in the U.S. airline industry.

\subsection{Motivation}
I investigate the U.S. airline industry as an empirical application.
In this industry, a market is defined as a city pair.
Each airline's entry decision involves deciding whether or not to serve a route between two cities.

In contrast to much of the existing literature, my model assumes that firms do not perfectly know their profits at the time they make their entry decisions due to payoff shocks.
In previous models, firms are assumed to observe their payoff shocks, and thus their profits when making decisions.
However, in my model, the uncertainty arising from unobserved payoff shocks induces firms to acquire information.

The uncertainty in profits, and thus the need for information acquisition, stems from three sources.
First, airlines do not know the exact demand for a route before selling flight tickets and therefore need to gather information to forecast it.
If a firm serves a route that attracts few passengers, costs can exceed revenue, leading to negative profits.
In such a case, the firm would be better off not entering the market.
Thus, gathering information about future demand plays a key role in their decision making.

Second, the operational feasibility at airports is important for airlines' entry decisions.
This includes factors such as airport infrastructures, slot availability, and ground handling capacities at a certain airport.
Assessing the operational feasibility might require airlines to acquire information from airport authorities.

Lastly, regulatory environments are a crucial factor in airlines' decisions.
To provide services between two cities, an airline needs to comply with the local aviation regulations and acquire information about them.
Moreover, taxes and fees can differ across airports, and the firms should be aware of these prior to the operation to make an optimal decision.

\subsection{Data}
My primary dataset on the airline industry comes from \citet{kline2016} \footnote{\citet{chen2018} and \citet{gu2025} use the same dataset and apply their models to airline entry game. In addition, \citet{kline2024} uses a similar dataset, except that the data are collected from the year 2023.}.
According to them, the data are collected from the second quarter of the 2010 Airline Origin and Destination Survey (DB1B).
The unit of observation is a market, and each observation contains information on airlines' market entry decisions and the market-specific and firm-specific explanatory variables.

The data contain $7,882$ markets indexed by $m=1,2,\ldots,7882$, and each market is defined as routes between two airports irrespective of stopping.
As in \citet{kline2016}, I consider the entry behavior of two players: (1) low-cost carriers ($LCC$) \footnote{According to \citet{kline2016}, AirTran, Allegiant Air, Frontier, JetBlue, Midwest Air, Southwest, Spirit, Sun Country, USA3000, and Virgin America are aggregated into one single player, $LCC$.}, and (2) other airlines ($OA$) \footnote{All airlines other than low-cost carriers are aggregated into one single player, $OA$.}.
In the following analysis, the players $LCC$ and $OA$ are treated as a single player, and $Y_{LCC}=1$ (respectively, $Y_{OA}=1)$ means that at least one airline firm in $LCC$ (respectively, $OA$) serves the market.
On the contrary, $Y_{LCC}=0$ (respectively, $Y_{OA}=0$) means that none of firms in $LCC$ (respectively, $OA$) serves the market.
The unconditional choice probabilities are $(\widehat{\Pr}(0,0), \widehat{\Pr}(1,0), \widehat{\Pr}(0,1), \widehat{\Pr}(1,1)) = (0.15, 0.07, 0.61, 0.16)$, where $\widehat{\Pr}(Y_{LCC},Y_{OA})$ is the sample frequency of the market outcome $(Y_{LCC},Y_{OA})$.

Two explanatory variables are used in my empirical analysis.
The first explanatory variable is a market- and firm-specific variable which has been used in many papers: market presence \citep{berry1992, ciliberto2009}.
The market presence, denoted by $Z_{im}$, is a proxy for a scale of operation of airline $i$ in market $m$ \citep{berry1992}.
Moreover, the market presence is a measure of the airline's market power and profitability since a large scale of operation may deter entry of other airlines \citep{gu2025}.
The market presence $Z_{im}$ varies across airlines and only enters airline $i$'s profits.
This variable satisfies the exclusion restriction and is essential for the payoff identification.
The market presence of airline $i$ at a given market $m$ is computed as follows.
For each airline and for each airport, the number of routes served by the airline from the airport is counted.
Next, the proportion of routes served from the airport is calculated by dividing the number of routes served by the airline by the total number of routes from that airport by all airlines.
The market presence for each airline at a given market $m$ is the average of these proportions at two endpoints of the market.
Since there are two players in my analysis, the market presence variable is the maximum of market presence for each airline in the group of $LCC$ and $OA$.

The second explanatory variable is market-specific variable: market size.
The market size, denoted by $X_m$, is defined by the population of the endpoints \footnote{The original data source, \citet{kline2016}, does not explicitly specify which of the two endpoints' population was used, nor does it clarify the source of the population data. However, \citet{gu2025}, which uses the same dataset, mentions that the population data is collected from the U.S. Census Bureau in 2010 and the market size is defined by the maximum of the two endpoints' population.}.
Both LCC and OA in the same market have the same market size, and the market size varies across the markets.
In the analysis, the unit of market size is billions of people.

Table \ref{tab:sumstat} reports the summary statistics of my data.
The decision variables are binary, and both the market presence and market size are continuous variables.
Apparently, the entry probability of OA is about 3 times higher than LCC.
Moreover, the airport presence of OA, a measure of operation scale, is about 7.9 times higher than LCC.
On average, markets that LCC serves tend to be larger than the markets that OA serves.

\begin{table}[ht]
  \centering
  \caption{Summary Statistics}
  \label{tab:sumstat}
  \begin{tabular}{lcccccc}
  \toprule
                                & Mean   & Std. Dev. & Median & Max & Min & Obs. \\
  \midrule
  $LCC$ Entry                   & 0.2322 & 0.4222    & 0      & 1   & 0   & 7,882 \\
  $OA$  Entry                   & 0.7742 & 0.4182    & 1      & 1   & 0   & 7,882 \\
  \midrule
  Market Presence               & 0.1986 & 0.1596 & 0.1000 & 0.5007 & 0 & 15,764 \\
  \addlinespace
  \quad Market Presence ($LCC$)       & 0.0447 & 0.0204 & 0.0429 & 0.1128 & 0 & 7,882 \\
  \quad Market Presence ($Y_{LLC}=1$) & 0.0676 & 0.0154 & 0.0677 & 0.1128 & 0.0149 & 1,830 \\ 
  \quad Market Presence ($Y_{LLC}=0$) & 0.0377 & 0.0161 & 0.0386 & 0.1053 & 0 & 6,052 \\
  \addlinespace
  \quad Market Presence ($OA$)        & 0.3525 & 0.0562 & 0.3581 & 0.5007 & 0.0349 & 7,882 \\ 
  \quad Market Presence ($Y_{OA}=1$)  & 0.3646 & 0.0437 & 0.3652 & 0.5007 & 0.1318 & 6,102 \\  
  \quad Market Presence ($Y_{OA}=0$)  & 0.3111 & 0.0723 & 0.3279 & 0.4655 & 0.0349 & 1,780 \\
  \midrule
  Market Size                   & 1,229,550 & 648,935 & 1,067,422 & 4,177,793 & 171,466 & 7,882 \\
  \addlinespace
  \quad Market Size ($Y_{LLC}=1$)     & 1,310,375 & 662,521 & 1,159,922 & 4,044,754 & 231,001 & 1,830 \\ 
  \quad Market Size ($Y_{LLC}=0$)     & 1,205,111 & 642,828 & 1,046,574 & 4,177,793 & 171,466 & 6,052 \\
  \addlinespace
  \quad Market Size ($Y_{OA}=1$)      & 1,286,312 & 646,809 & 1,131,214 & 4,177,793 & 173,002 & 6,102 \\ 
  \quad Market Size ($Y_{OA}=0$)      & 1,034,966 & 617,953 & 874,540   & 4,177,793 & 171,466 & 1,780 \\
  \bottomrule
\end{tabular}
\end{table}

\subsection{Preliminary Analysis}
Before estimating the structural model, I conduct a preliminary analysis using simple Probit regressions to examine the data patterns.
The Probit regression equation is given by the following equation:
  \begin{equation}
  Y_{im} = \mathds{1} \left\{ \beta^{cons} + \beta^{pres} Z_{im} + \beta^{size} X_{m} + \delta Y_{jm}+ \varepsilon_{im} \geq 0 \right\},
\end{equation}
where the error term $\varepsilon_{im}$ follows the standard normal distribution.
In simple Probit regressions, each observation represents a specific player within a particular market.
The dependent variable is entry decisions, $Y_{im} \in\{0,1\}$.
I regress the binary decisions on firm-specific and market-specific variables.

\begin{table}[ht]
  \centering
  \caption{Estimation results from Probit regressions}
  \label{tab:prelim}
\begin{tabular}{lcccc}
\toprule
                   &   (1)    &   (2)    &   (3)    \\
\midrule
Market Presence    & 5.6688   & 5.7304   & 5.3937   \\
                   &(0.0757)  &(0.0768)  &(0.0899)  \\
\addlinespace 
Market Size        &          & 0.2396   & 0.2582   \\
                   &          &(0.0182)  &(0.0185)  \\
\addlinespace 
Competitive Effect &          &          & -0.1936  \\
                   &          &          &(0.0277)  \\
\addlinespace 
Constant           & -1.1020  & -1.4030  & -1.2632  \\
                   &(0.0186)  &(0.0300)  &(0.0358)  \\
\midrule
Log Likelihood     & -7648.11 & -7560.80 & -7536.69 \\
Observations       & 15,764   & 15,764   & 15,764   \\
\bottomrule
\end{tabular}
\smallskip
\begin{tablenotes}
\footnotesize
\item Notes: This table reports the estimation results of the Probit regressions.
  Each observation represents a specific player within a particular market.
  Standard errors are in the parentheses.
\end{tablenotes}
\end{table}

Table \ref{tab:prelim} reports the estimation results from the Probit regressions.
The coefficients indicate the impact of each variable on the underlying latent profit for the entry decision.
The estimation results are reasonable in two ways.
First, both market presence and market size are positively correlated with the entry decisions.
This result is consistent with the intuition that an airline with a large operational scale in a large demand market is more likely to serve the route.
Second, the competitive effect is negatively correlated with the entry decisions, indicating that a rival's entry reduces the player's profits.

Although the preliminary analysis result is consistent with the intuition on airlines' entry decisions, the estimation results may be biased.
In the above analysis, the airlines' profits are assumed to be perfectly observable to airlines.
As I have discussed in the previous section, the profits may not be perfectly known to airlines due to uncertainty, operational feasibility, and regulatory environments.
I model the structural entry game with costly information acquisition to properly incorporate strategic interactions in the next subsection.

\subsection{Structural Entry Game Analysis}
I develop a standard parametric entry game that builds on the existing literature by adding costly information acquisition.
The two firms are $LCC$ and $OA$, and they choose a binary action $Y_{im} \in \left\{ 0,1 \right\}$ where $Y_{im}=1$ if firm $i$ enters a market $m$ and $Y_{im}=0$ if $i$ stays out.
The profit function for firm $i$ is specified as follows:
\begin{equation}
  \label{eq:sec5-structuralpayoff}
  U_{im}(Y_{im},Y_{jm}, \varepsilon_{im}; X_{m}, Z_{im}, \theta) = Y_{im} \left(\beta_{i}^{cons} + \beta_{i}^{pres} Z_{im}+ \beta_{i}^{size} X_{m} + \delta_{i} Y_{jm} + \varepsilon_{im} \right), 
\end{equation}
where $\theta$ is the vector of parameters of interest.

Firm $i$'s profit from entering the market is $\beta_{i}^{cons} + \beta_{i}^{pres} Z_{im}+ \beta_{im}^{size} X_{m} + \delta_{i} Y_{jm} + \varepsilon_{im}$ where $Z_{im}$ represents the market presence, $X_m$ represents the market size, $Y_{jm}$ represents rival's entry, and $\varepsilon_{im}$ is the payoff shock.
The payoff shock $\varepsilon_{im}$ is unobserved by firm $i$ (and also by rival $j$ and econometricians) when the firm $i$ makes entry decision.
The profit from staying out is normalized to zero.
I assume that each player's prior belief about $\varepsilon_{im}$ follows an i.i.d. standard normal distribution.
This is because the normal distribution describes the maximum entropy among all distributions with the same variance.
Given this prior belief, firms make entry decisions under the condition of maximum uncertainty.

I estimate the parameters by using the two-step estimation method due to multiple equilibria.
Under the multiple equilibria, a vector of parameters $\theta$ corresponds to multiple choice probabilities.
Since I do not know which choice probability is selected by the firms in the data, I estimate the choice probabilities in the first step.
In the second step, given the choice probabilities, I estimate the parameters by maximizing the pseudo log-likelihood function.

In the first step, I nonparametrically estimate the conditional choice probabilities.
I denote the conditional choice probability of entering the market as $P_{im}(x_m, \mathbf{z}_m) = P_{im}(Y_{im}=1 | x_m, \mathbf{z}_{m})$.
Since the entry decision is binary, the conditional choice probability of staying out can be expressed as $1 - P_{im}(x_{m}, \mathbf{z}_{m})$.
The conditional choice probabilities are estimated using sieve method \citep{chen2007}.
Let $q^{\kappa}(x_m, \mathbf{z}_m) = (q_{1}(x_m, \mathbf{z}_m), \ldots, q_{\kappa}(x_m, \mathbf{z}_m))^{\top}$ be the $\kappa\times1$ vector of basis functions with degree $\kappa$.
The sieve estimator of the conditional choice probabilities are:
\begin{equation}
  \label{eq:sec5-sieve}
  \widehat{P}_{im}(x_m, \mathbf{z}_m) = \frac{\exp(q^{\kappa}(x_m, \mathbf{z}_m)' \hat{\gamma})}{1 + \exp(q^{\kappa}(x_m, \mathbf{z}_m)' \hat{\gamma})},
\end{equation}
where $\hat{\gamma}$ is the maximizer of the log likelihood function
\begin{equation}
  \label{eq:sec5-sievell}
  \hat{\gamma} = \arg \max_{\gamma} \sum_{m=1}^M \left[ \mathds{1} \left\{Y_{im}=1 \right\} q^{\kappa}(x_m, \mathbf{z}_{m})' \gamma - \ln \left(1+\exp(q^{\kappa}(x_m, \mathbf{z}_{m})' \gamma) \right) \right].
\end{equation}
In the first step estimation, I set the degree $\kappa=3$.
The sieve estimator $\widehat{P}_{im}(x_m, \mathbf{z}_m)$ will converge to the true conditional choice probability $P_{im}(x_m, \mathbf{z}_m)$.

In the second step estimation, I take as given the first step estimates $\widehat{P}_{im}(x_m, \mathbf{z}_m)$ of the conditional choice probabilities.
I then construct the pseudo log-likelihood function and find the maximizer of the pseudo log-likelihood function.
For a given value of $(\mathbf{x}, \mathbf{z})$, for a given $y = (y_{LCC}, y_{OA})$, I estimate the parameters $\theta$ by maximizing the following pseudo log-likelihood function:
\begin{equation}
  \label{eq:sec5-pmle-obj}
  \ell(\theta, \widehat{\mathbf{P}}) =  \sum_{m=1}^M \sum_{i} y_{im} \ln \Psi_{i} \left( \widehat{\mathbf{P}}_{m}(x_m, \mathbf{z}_m), x_m, \mathbf{z}_m, \theta  \right) + (1-y_{im}) \ln \left( 1 - \Psi_{i} \left( \widehat{\mathbf{P}}_{m}(x_m, \mathbf{z}_m), x_m, \mathbf{z}_m, \theta  \right) \right),
\end{equation}
where $\Psi_i(\cdot)$ is the equilibrium condition defined by
\begin{align}
  \label{eq:sec5-pmle-psi}
  &\Psi_{i} \left( \widehat{\mathbf{P}}_{m}(x_m, \mathbf{z}_m), x_m, \mathbf{z}_m, \theta  \right) \nonumber \\
  &\quad = \int_{\varepsilon_i} \frac{\widehat{P}_{im}(x_m, \mathbf{z}_{m}) \exp\left(\beta_i^{cons} + \beta_i^{pres}z_{im} + \beta_{i}^{size}x_m + \delta_i\widehat{P}_{j}(x_m, \mathbf{z}_{m}) + \varepsilon_{im} \right)}{\widehat{P}_{im}(x_m, \mathbf{z}_{m}) \exp\left(\beta_i^{cons} + \beta_i^{pres}z_{im} + \beta_{i}^{size}x_m + \delta_i\widehat{P}_{jm}(x_m, \mathbf{z}_{m}) + \varepsilon_{im} \right) + 1 - \widehat{P}_{im}(x_m, \mathbf{z}_m) } \mathrm{d}\Phi(\varepsilon_{im}). \nonumber
\end{align}

\subsection{Estimation Result}
Table \ref{tab:estimgame} presents the estimation results under the assumption that players acquire costly information.
For both $LCC$ and $OA$, the coefficients on market presence and market size are positive and statistically significant.
This means that markets with existing airline presence and larger market sizes generally offer higher expected profits for entry.
The constant terms for both $LCC$ and $OA$, which capture the average profit from entry independent of the observed firm and market characteristics, are negative and statistically significant.
The negative constant term can be interpreted as the entry barrier faced by firms, and the results indicate that this barrier is more substantial for $LCC$ than $OA$.

The competitive effects are negative for both $LCC$ and $OA$, consistent with the expectation that increased competition reduces a firm's entry profit.
However, the coefficient for $LCC$ is not statistically significant.
This suggests that, conditional on other variables, the entry of $OA$ does not influence $LCC$'s entry decision.
By contrast, $LCC$'s entry decision has a statistically significant negative effect on $OA$'s entry decision.
In the context of costly information acquisition, this implies that $LCC$ relies more heavily on its own payoff uncertainty, making the acquisition of precise information more critical for its decision.

\begin{table}[ht]
\centering
\caption{Estimation results of the entry game with costly information acquisition}
\label{tab:estimgame}
\begin{tabular}{lcc}
\toprule
                  & $LCC$    & $OA$     \\
\midrule
Market Presence   & 16.4100  & 2.5118  \\
                  & (0.5552) & (0.1185) \\
\addlinespace 
Market Size       & 0.0407   & 0.0853   \\
                  & (0.0124) & (0.0104) \\
\addlinespace 
Competitive Effect& -0.0301  & -0.1963  \\
                  & (0.0651) & (0.0517) \\
\addlinespace 
Constant          & -1.0847  & -0.7013  \\
                  & (0.0653) & (0.0469) \\
\midrule
Log Likelihood    & -5839.21 &          \\
Observations      & 7882     &          \\
\bottomrule
\end{tabular}
  \begin{tablenotes}
    \footnotesize
    \item Notes: This table reports the estimation results of the model with costly information acquisition.
    The standard errors are computed using the bootstrap method with 1,000 replications.
    The standard errors are in parentheses.
  \end{tablenotes}
\end{table}

\subsection{Information Acquisition}
I can compute the amount of information acquired by each firm in each market using the estimated parameters.
The amount of information acquired by firm $i$ in market $m$ is given by
\begin{equation}
  I_{im}(F_{im}, P_{im}(\mathbf{x}_{m}, \mathbf{z}_{m})) = H(F_{im}(\cdot)) - \mathrm{E} \left[ H(P_{im}(\cdot| y_{im}; \mathbf{x}_{m}, \mathbf{z}_{m})) \right],
\end{equation}
where $P_{im}(\mathbf{x}_{m}, \mathbf{z}_{m})$ is the equilibrium choice probability and $P_{im}(\cdot|, y_{im};\mathbf{x}_{m}, \mathbf{z}_{m})$ is the posterior beliefs.

Table \ref{tab:infoacqui} reports the mean, median, and maximum of the amount of information acquired by $LCC$ and $OA$.
The mean of the amount of information acquired by $LCC$ is from $0.0363$, while that by $OA$ is $0.0543$.
The $LCC$ tends to acquire less information than $OA$.
The same pattern is observed in the median and maximum of the amount of information acquired.
The median and the maximum of the amount of information acquired by $LCC$ are $0.0250$ and $0.1198$, respectively, which are lower than those of $OA$.
These results indicate that $LCC$ tends to acquire less information than $OA$, which may be due to their business model that focuses on cost efficiency.

\begin{table}[h]
    \centering
    \caption{The Amount of Information Acquired}
    \label{tab:infoacqui}
    \begin{tabular}{lcc}
    \toprule
            & LCC    & OA     \\
    \midrule
    Mean    & 0.0363 & 0.0543 \\
    \addlinespace 
    Median  & 0.0250 & 0.0546 \\
    \addlinespace 
    Max     & 0.1198 & 0.1330 \\
    \bottomrule
    \end{tabular}
\end{table}

Figure \ref{fig:infoacqui} presents the distribution of the amount of information acquired by $LCC$ and $OA$, along with fitted normal density functions.
The distribution of the amount of information acquired by $OA$ appears to have a higher mean and a smaller variance than that of $LCC$.
This indicates that $OA$ tends to acquire more information than $LCC$ and suggests systematic differences in the role of information across the two firms.

\begin{figure}[h!]
    \centering
    \includegraphics[width=0.8\textwidth]{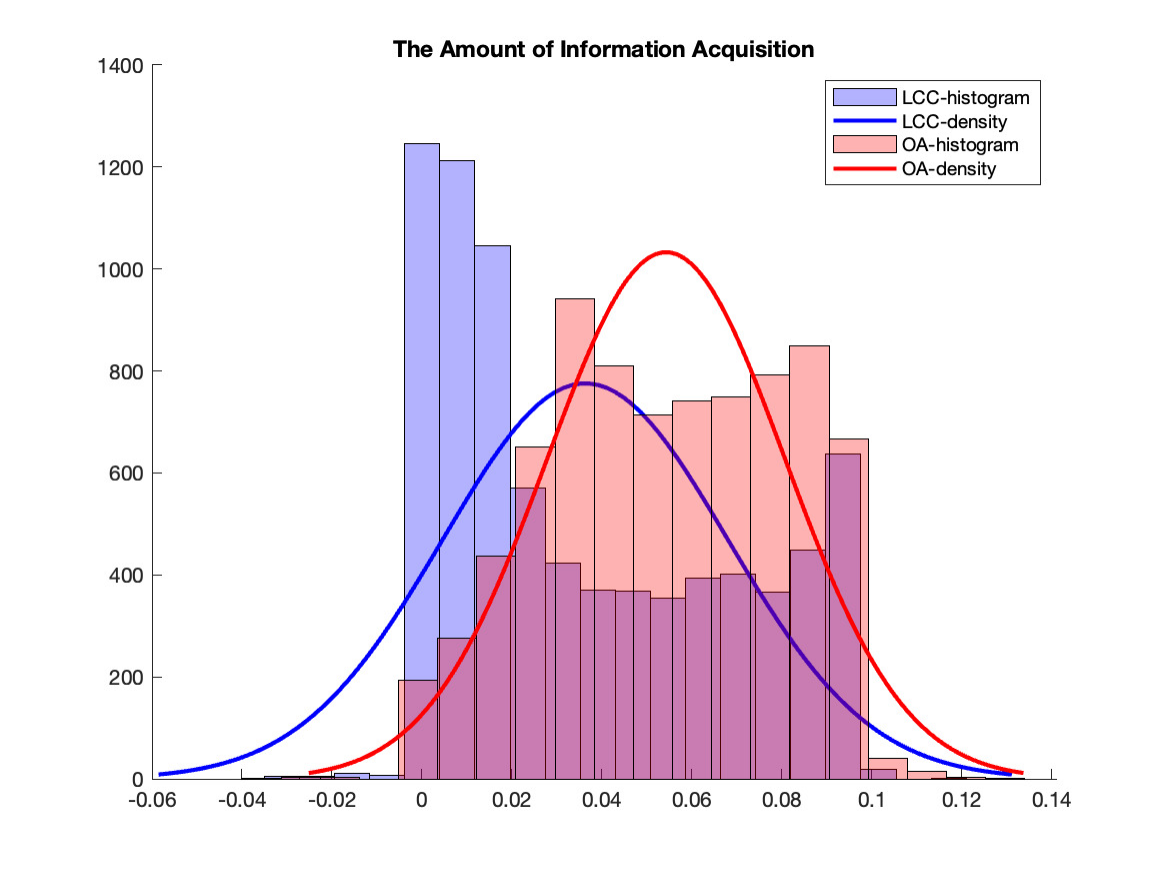}
    \caption{The amount of information acquired}
    \label{fig:infoacqui}
    \begin{tablenotes}
    \footnotesize
    \item Notes: This figure shows the histogram of the estimated amount of information acquired by $LCC$ and $OA$. 
    The solid lines represent the fitted normal density functions.
    The negative values of the amount of information acquired are due to the numerical approximation errors.
  \end{tablenotes}
\end{figure}

\subsubsection{Comparison to Perfectly Private Information Structure}
I compare the model with costly information acquisition to the model without such costs.
When the cost of information is zero, each player acquires full information about the exact value of their own payoff shock $\varepsilon_{i}$.
This corresponds to the perfectly private information structure studied in \citet{seim2006, bajari2010, aradillas2010}.
The key distinction between the two models lies in whether firms observe their own payoff shock when making entry decisions.
In my model with costly information acquisition, firms do not observe payoff shocks prior to entry decisions.
In contrast, under the perfectly private information structure, firms perfectly observe their own payoff shock.
The estimation results for the perfectly private information structure model are reported in Table \ref{tab:estimgame-perfpriv}.

In the model without information acquisition, the estimated coefficients on market presence and market size are substantially larger for both $LCC$ and $OA$ compared to the model with costly information acquisition.
This suggests that when firms observe their private payoff shocks, the observed market characteristics become more influential in explaining their entry decisions.
The constant terms also become significantly more negative in the perfectly private information model, implying a larger barrier of entry.

The estimated competitive effects differ across the two models.
In the model with costly information acquisition, the absolute value of estimates of competitive effect for $LCC$ is smaller than that for $LCC$ in perfectly private information model.
However, the competitive effect remains insignificant for $LCC$.
The estimates for $OA$ show similar results as in the perfectly private information model.

All the coefficients, except for the $LCC$'s competitive effect, increase in the magnitude of coefficients from the perfectly private information model, as compared to the costly information acquisition model.
This suggests that when firms are modeled with perfectly private information structure, the estimation may attribute a greater influence to observed firm and market characteristics.
When firms are assumed to perfectly observe their own payoff shocks, conditional choice probabilities are more affected by observable variables, such as market presence and market size.
As a result, the model may overestimate the influence of these variables.

Finally, the lower log likelihood for the perfectly private information model indicates that the model with costly information acquisition provides a better statistical fit to the observed data.
This finding supports the idea that firms' actual decision making likely incorporates a more sophisticated, endogenous information strategy.

\begin{table}[ht]
\centering
\caption{Estimation results of the entry game without information acquisition}
\label{tab:estimgame-perfpriv}
\begin{tabular}{lcc}
\toprule
                  & $LCC$    & $OA$     \\
\midrule
Market Presence   & 62.9220  & 9.6254   \\
                  & (1.4240) & (0.3466) \\
\addlinespace 
Market Size       & 0.1141   & 0.3315   \\
                  & (0.0331) & (0.0302) \\
\addlinespace 
Competitive Effect& -0.0091  & -0.3254  \\
                  & (0.1057) & (0.0610) \\
\addlinespace 
Constant          & -4.1228  & -2.8645  \\
                  & (0.1245) & (0.1242) \\
\midrule
Log Likelihood    & -5986.84 &          \\
Observations      & 7882     &          \\
\bottomrule
\end{tabular}
\begin{tablenotes}
  \footnotesize
  \item Notes: This table reports the estimation results of the perfectly private information model.
    In this information structure, each firm only observe their own payoff shock $\varepsilon_{i}$, but not rival's payoff shock $\varepsilon_{j}$.
    The standard errors are in the parentheses.
\end{tablenotes}
\end{table}

\section{Conclusion}
\label{sec:conc}
This paper has proposed a novel econometric framework for static discrete games with costly information acquisition.
In contrast to the existing literature, in the model, players acquire information about their own payoffs, but acquiring the information incurs costs, which is proportional to the mutual information.
I characterized the Nash equilibrium conditions and established equilibrium existence and uniqueness condition.
I showed that the model primitives are identified for  for both parametric and semiparametric settings, and I presented the identification results via Monte Carlo experiments.
Finally, I applied the proposed framework to U.S airline market entry decisions and found that the low-cost carriers acquire less information about profits compared to other airlines, which may be due to their business model that focuses on cost efficiency.


\newpage
\bibliographystyle{chicago}
\bibliography{bibliography}

\newpage
\begin{appendix}
\section{Proofs}
\subsection{Proof of Lemma 1}
\begin{proof}
Let $\sigma_{i}(\tau_{i}^{\mathbf{X}, \mathbf{Z}})$ be player $i$'s optimal action decision rule after receiving a signal $\tau_{i}^{\mathbf{X}, \mathbf{Z}} \in \mathcal{T}_{i}^{\mathbf{X}, \mathbf{Z}}$.
Suppose to the contrary that player $i$ plays not only pure strategies but also mixed strategies.
Define the following sets
    \begin{align*}
        \mathcal{T}_{i0} &:= \{ \tau_{i}^{\mathbf{X}, \mathbf{Z}} \in \mathcal{T}_{i}^{\mathbf{X}, \mathbf{Z}} \,|\, \sigma_{i} (\tau_{i}^{\mathbf{X}, \mathbf{Z}}) = 0 \}, \\
        \mathcal{T}_{i1} &:= \{ \tau_{i}^{\mathbf{X}, \mathbf{Z}} \in \mathcal{T}_{i}^{\mathbf{X}, \mathbf{Z}} \, | \, \sigma_{i} (\tau_{i}^{\mathbf{X}, \mathbf{Z}}) = 1 \}, \\
        \mathcal{T}_{i,mixed} &:= \{ \tau_{i}^{\mathbf{X}, \mathbf{Z}} \in \mathcal{T}_{i}^{\mathbf{X}, \mathbf{Z}} \,|\, \sigma_{i} (\tau_{i}^{\mathbf{X}, \mathbf{Z}}) \in (0,1) \}.
    \end{align*}
The three sets form a partition of player $i$'s signal space $\mathcal{T}_{i}^{\mathbf{X}, \mathbf{Z}}$.
Due to the information cost, the player does not have an incentive to investigate signal realizations that induce the same action.
For instance, the player does not distinguish $\tau'_{i} \neq \tau_{i}$ such that $\tau_{i}, \tau'_{i} \in \mathcal{T}_{i1}$ and $\sigma_{i}(\tau'_{i}) = \sigma_{i}(\tau_{i}) = 1$.
The effort to process the details of signals incurs information costs but gives no extra benefits.
Thus, player $i$'s signal space is equivalent to a set with at most three elements.

Given the optimal information structure $S_{i}^{\mathbf{X}, \mathbf{Z}}$ and the optimal action decision rule $\sigma_{i}(\cdot)$, the expected payoff under $S_{i}^{\mathbf{X}, \mathbf{Z}}$ excluding the information cost is 
    \begin{equation*}
        \int_{\varepsilon_{i}} \int_{\tau_{i}^{\mathbf{X}, \mathbf{Z}}} U_{i}(Y_{i}, Y_{-i}, \varepsilon_{i}; \mathbf{X}, \mathbf{Z}_{i}) \sigma_{i}(\tau_{i}^{\mathbf{X}, \mathbf{Z}}) \, P_{i}^{\mathbf{X}, \mathbf{Z}}(\tau_{i}^{\mathbf{X}, \mathbf{Z}} \,|\,\varepsilon_{i}) \mathrm{d} \tau_{i}^{\mathbf{X}, \mathbf{Z}} \mathrm{d} F_{i}(\varepsilon_{i}).
    \end{equation*}

Next, consider a new information structure $\Tilde{S}_{i}^{\mathbf{X}, \mathbf{Z}} = (\Tilde{\mathcal{T}}_{i}^{\mathbf{X}, \mathbf{Z}}, \Tilde{P}_{i}^{\mathbf{X}, \mathbf{Z}})$ such that $\Tilde{\mathcal{T}}_{i}^{\mathbf{X}, \mathbf{Z}} = \{t_{0}, t_{1}\}$. 
Let the optimal action decision rule $\Tilde{\sigma}_{i}(\cdot)$ associated with the new information structure $\Tilde{S}_{i}^{\mathbf{X}, \mathbf{Z}}$ be 
    \begin{align*}
        \Tilde{\sigma}_{i}(t_{0}) &= 0\\
        \Tilde{\sigma}_{i}(t_{1}) &= 1.
    \end{align*}

The distribution of the signals are defined by
    \begin{align*}
        \Tilde{P}_{i}^{\mathbf{X}, \mathbf{Z}}(t_{0} \,|\, \varepsilon_{i}) &= \int_{\mathcal{T}_{i0} \cup \mathcal{T}_{i,mixed}} P_{i}^{\mathbf{X}, \mathbf{Z}}(\tau_{i}^{\mathbf{X}, \mathbf{Z}} \,|\, \varepsilon_{i}) \mathrm{d} \tau_{i}^{\mathbf{X}, \mathbf{Z}} \\
        \Tilde{P}_{i}^{\mathbf{X}, \mathbf{Z}}(t_{1} \,|\, \varepsilon_{i}) &= \int_{\mathcal{T}_{i1}} P_{i}^{\mathbf{X}, \mathbf{Z}}(\tau_{i}^{\mathbf{X}, \mathbf{Z}} \,|\, \varepsilon_{i}) \mathrm{d} \tau_{i}^{\mathbf{X}, \mathbf{Z}}
    \end{align*}
for any payoff shock $\varepsilon_{i} \in \mathcal{E}_{i}$.
The expected payoff under $\Tilde{S}_{i}^{\mathbf{X}, \mathbf{Z}}$ excluding the information cost is identical to that under $S_{i}$ since
    \begin{align*}
        & \int_{\varepsilon_{i}} U_{i}(Y_{i},Y_{-i},\varepsilon_{i};\mathbf{X}, \mathbf{Z}_{i}) \Tilde{\sigma}_{i}(t_{0}) \, \Tilde{P}_{i}^{\mathbf{X}, \mathbf{Z}}(t_{0} \,|\,\varepsilon_{i}) \mathrm{d} F_{i}(\varepsilon_{i}) + \int_{\varepsilon_{i}} U_{i}(Y_{i}, Y_{-i},\varepsilon_{i};\mathbf{X}, \mathbf{Z}_{i}) \Tilde{\sigma}_{i}(t_{1}) \, \Tilde{P}_{i}^{\mathbf{X}, \mathbf{Z}}(t_{1} \,|\,\varepsilon_{i}) \mathrm{d} F_{i}(\varepsilon_{i}) \\
        &= \int_{\varepsilon_{i}} \int_{\mathcal{T}_{i0} \cup \mathcal{T}_{i,mixed}} U_{i}(Y_{i}, Y_{-i}, \varepsilon_{i};\mathbf{X}, \mathbf{Z}_{i}) \sigma_{i}(\tau_{i}^{\mathbf{X}, \mathbf{Z}}) \, P_{i}^{\mathbf{X}, \mathbf{Z}}(\tau_{i}^{\mathbf{X}, \mathbf{Z}} \,|\,\varepsilon_{i}) \mathrm{d} \tau_{i}^{\mathbf{X}, \mathbf{Z}} \mathrm{d} F_{i}(\varepsilon_{i}) \\
        &\qquad + \int_{\varepsilon_{i}} \int_{\mathcal{T}_{i1}} U_{i}(Y_{i}, Y_{-i}, \varepsilon_{i}; \mathbf{X}, \mathbf{Z}_{i}) \sigma_{i}(\tau_{i}^{\mathbf{X}, \mathbf{Z}}) \, P_{i}^{\mathbf{X}, \mathbf{Z}}(\tau_{i}^{\mathbf{X}, \mathbf{Z}} \,|\,\varepsilon_{i}) \mathrm{d} \tau_{i}^{\mathbf{X}, \mathbf{Z}} \mathrm{d} F_{i}(\varepsilon_{i}) \\
        &= \int_{\varepsilon_{i}} \int_{\tau_{i}^{\mathbf{X}, \mathbf{Z}}} U_{i}(Y_{i}, Y_{-i}, \varepsilon_{i};\mathbf{X}, \mathbf{Z}_{i}) \sigma_{i}(\tau_{i}^{\mathbf{X}, \mathbf{Z}}) \, P_{i}^{\mathbf{X}, \mathbf{Z}}(\tau_{i}^{\mathbf{X}, \mathbf{Z}} \,|\,\varepsilon_{i}) \mathrm{d} \tau_{i}^{\mathbf{X}, \mathbf{Z}} \mathrm{d} F_{i}(\varepsilon_{i})
    \end{align*}

Notice that the information structure $\Tilde{S}_{i}^{\mathbf{X}, \mathbf{Z}}$ is less informative than $S_{i}^{\mathbf{X}, \mathbf{Z}}$ if $\mathcal{T}_{i,mixed}$ is nonempty or $|\mathcal{T}_{i}| > 2$.
If $\mathcal{T}_{i,mixed}$ is nonempty or $|\mathcal{T}_{i}| > 2$, then $\Tilde{S}_{i}^{\mathbf{X}, \mathbf{Z}}$ does not require player $i$ to process signal realizations within $\mathcal{T}_{i0}$ and $\mathcal{T}_{i,mixed}$.
The mutual information cost function has a property that the more informative it is, the more it costs.
By this property, $\Tilde{S}_{i}^{\mathbf{X}, \mathbf{Z}}$ is cheaper than $S_{i}^{\mathbf{X}, \mathbf{Z}}$, and the player's net expected payoff is greater under the new information structure $\Tilde{S}_{i}^{\mathbf{X}, \mathbf{Z}}$.
This contradicts the optimal information structure $S_{i}^{\mathbf{X}, \mathbf{Z}}$.
Therefore, the player does not choose mixed strategy and the cardinality of the signal space is at most two.
\end{proof}

\subsection{Proof of Proposition 1.}
Let the rival's information strategy be $P_{j}^{\mathbf{X}, \mathbf{Z}}(Y_{j} | \varepsilon_{j})$.
Define the deterministic expected payoff \eqref{eq:deterministicpayoff} by
\begin{equation*}
    EU_{i}(\mathbf{X}, \mathbf{Z}) = \pi_{i}(\mathbf{X}, \mathbf{Z}_{i}) + \delta_{i}(\mathbf{X}, \mathbf{Z}_{i}) P_{j}^{\mathbf{X}, \mathbf{Z}}(Y_{j} | \varepsilon_{j}).
\end{equation*}
I construct the Lagrangian for the expected payoff maximization problem:
    \begin{align*}
        \mathcal{L}(P_{i}) &= \int_{\varepsilon_{i}} \sum_{Y_{i} \in \mathcal{Y}_{i}} \left(EU_{i}(\mathbf{X}, \mathbf{Z}) + \varepsilon_{i} \right)P_{i}^{\mathbf{X}, \mathbf{Z}}(Y_{i} | \varepsilon_{i})  \,\mathrm{d} F(\varepsilon_{i}) \\
        &\; -\lambda_{i} \left[ \int_{\varepsilon_{i}} \left( \sum_{Y_{i} \in \mathcal{Y}_{i}} P_{i}^{\mathbf{X}, \mathbf{Z}}(Y_{i} | \varepsilon_{i}) \log P_{i}^{\mathbf{X}, \mathbf{Z}}(Y_{i} | \varepsilon_{i}) \right) \mathrm{d}F_{i} (\varepsilon_{i}) - \sum_{Y_{i} \in \mathcal{Y}_{i}} P_{i}^{\mathbf{X}, \mathbf{Z}}(Y_{i}) \log P_{i}^{\mathbf{X}, \mathbf{Z}}(Y_{i}) \right] \\
        &\; + \sum_{Y_{i} \in \mathcal{Y}_{i}} \int_{\varepsilon_{i}} \xi (\varepsilon_{i}) P_{i}^{\mathbf{X}, \mathbf{Z}}(Y_{i} | \varepsilon_{i}) \,\mathrm{d}F_{i}(\varepsilon_{i}) - \int_{\varepsilon_{i}} \mu(\varepsilon_{i}) \left( \sum_{Y_{i} \in \mathcal{Y}_{i}} P_{i}^{\mathbf{X}, \mathbf{Z}}(Y_{i} | \varepsilon_{i}) - 1 \right) \,\mathrm{d}F_{i}(\varepsilon_{i})
    \end{align*}
    where 
    \begin{equation*}
        P_{i}^{\mathbf{X}, \mathbf{Z}}(Y_{i}) = \int_{\varepsilon_{i}} P_{i}^{\mathbf{X}, \mathbf{Z}}(Y_{i} | \varepsilon_{i}) \,\mathrm{d}F_{i}(\varepsilon_{i}),
    \end{equation*}
    $\xi (\varepsilon_{i}) \geq 0$ is the Lagrangian multiplier for the constraint:
    \begin{equation*}
        P_{i}^{\mathbf{X}, \mathbf{Z}}(Y_{i} | \varepsilon_{i}) \geq 0, \quad \forall \varepsilon_{i} \in \mathcal{E}_{i}
    \end{equation*}
    and $\mu(\varepsilon_{i})$ is the Lagrangian multiplier for the constraint:
    \begin{equation*}
        \sum_{Y_{i}}P_{i}^{\mathbf{X}, \mathbf{Z}}(Y_{i} | \varepsilon_{i}) = 1, \quad \forall \varepsilon_{i} \in \mathcal{E}_{i}
    \end{equation*}
    If $\int_{\varepsilon_{i}} P_{i}^{\mathbf{X}, \mathbf{Z}}(Y_{i} | \varepsilon_{i}) \,\mathrm{d}F_{i}(\varepsilon_{i}) >0$, the first-order condition with respect to $P_{i}^{\mathbf{X}, \mathbf{Z}}(Y_{i} | \varepsilon_{i})$ is:
    \begin{equation*}
        \left(EU_{i}(\mathbf{X}, \mathbf{Z}) + \varepsilon_{i} \right) - \lambda_{i} \left[ \log P_{i}^{\mathbf{X}, \mathbf{Z}}(Y_{i} | \varepsilon_{i}) + 1 - \log \left( P_{i}^{\mathbf{X}, \mathbf{Z}}(Y_{i}) \right) -1 \right] + \xi (\varepsilon_{i}) - \mu(\varepsilon_{i}) = 0
    \end{equation*}

    Since $\int_{\varepsilon_{i}} P_{i}^{\mathbf{X}, \mathbf{Z}}(Y_{i} | \varepsilon_{i}) \,\mathrm{d}F_{i}(\varepsilon_{i}) >0$ implies $P_{i}^{\mathbf{X}, \mathbf{Z}}(Y_{i} | \varepsilon_{i}) > 0$, the complementary-slackness condition gives $\xi(\varepsilon_{i}) = 0$ and $\mu(\varepsilon_{i}) > 0$.
    Thus I have
    \begin{equation*}
        P_{i}^{\mathbf{X}, \mathbf{Z}}(Y_{i} | \varepsilon_{i}) = P_{i}^{\mathbf{X}, \mathbf{Z}}(Y_{i}) \exp\left( \frac{\left(EU_{i}(\mathbf{X}, \mathbf{Z}) + \varepsilon_{i} \right) - \mu(\varepsilon_{i})}{\lambda_{i}} \right)
    \end{equation*}
    
    Plugging the above expression for $P_{i}^{\mathbf{X}, \mathbf{Z}}(Y_{i} | \varepsilon_{i})$ into the constraint $\sum_{Y_{i}} P_{i}^{\mathbf{X}, \mathbf{Z}}(Y_{i} | \varepsilon_{i}) = 1$ gives 
    \begin{equation}
        P_{i}^{\mathbf{X}, \mathbf{Z}}(Y_{i} | \varepsilon_{i}) = \frac{P_{i}^{\mathbf{X}, \mathbf{Z}}(Y_{i}) \exp\left(EU_{i}(\mathbf{X}, \mathbf{Z}) + \varepsilon_{i}\right)^{1/\lambda_{i}}}{\sum_{Y_{i}'}P_{i}^{\mathbf{X}, \mathbf{Z}}(Y_{i}')\exp\left( EU_{i}(\mathbf{X}, \mathbf{Z}) + \varepsilon_{i}\right)^{1/\lambda_{i}}}
    \end{equation}
    
\subsection{Proof of Proposition 2.}
\begin{proof}
    Let $\mathbf{P}(\mathbf{X}, \mathbf{Z}) \equiv (P_{1}(\mathbf{X}, \mathbf{Z}_{1}, \mathbf{Z}_{2}), P_{2}(\mathbf{X}, \mathbf{Z}_{2}, \mathbf{Z}_{1})) \in [0,1]^{2}$.
    I define the fixed point mapping $\Phi:[0,1]^{2} \to [0,1]^{2}$
        \begin{equation*}
            \Phi(\mathbf{P}(\mathbf{X}, \mathbf{Z}),\mathbf{X}, \mathbf{Z}) = 
            \begin{pmatrix}
                \int_{\varepsilon_{1}} \frac{ P_{1}(\mathbf{X}, \mathbf{Z})\exp\left(EU_{1}(\mathbf{X}, \mathbf{Z}) + \varepsilon_{1}\right)^{1/\lambda_{1}}}{P_{1}(\mathbf{X}, \mathbf{Z})\exp\left(EU_{1}(\mathbf{X}, \mathbf{Z}) + \varepsilon_{1}\right)^{1/\lambda_{1}} + 1 - P_{1}(\mathbf{X}, \mathbf{Z})} \mathrm{d}F_{1}(\varepsilon_{1}) \\
                \int_{\varepsilon_{2}} \frac{ P_{2}(\mathbf{X}, \mathbf{Z})\exp\left(EU_{2}(\mathbf{X}, \mathbf{Z}) + \varepsilon_{2}\right)^{1/\lambda_{2}}}{P_{2}(\mathbf{X}, \mathbf{Z})\exp\left(EU_{2}(\mathbf{X}, \mathbf{Z}) + \varepsilon_{2}\right)^{1/\lambda_{2}} + 1 - P_{2}(\mathbf{X}, \mathbf{Z})} \mathrm{d}F(\varepsilon_{2})
            \end{pmatrix}.
        \end{equation*}
    Given Assumption 1, $\Phi$ satisfies the conditions for Brouwer's Fixed Point Theorem.
    Therefore, there exists a fixed point $(P_{1}(\mathbf{X}, \mathbf{Z}_{1}, \mathbf{Z}_{2}), P_{2}(\mathbf{X}, \mathbf{Z}_{2}, \mathbf{Z}_{1})) \equiv\mathbf{P}(\mathbf{X}, \mathbf{Z}) \in (0,1)^{2}$ such that solves  
        \begin{equation*}
            \Phi(\mathbf{P}(\mathbf{X}, \mathbf{Z}),\mathbf{X}, \mathbf{Z}) = \mathbf{P}(\mathbf{X}, \mathbf{Z}).
        \end{equation*}
\end{proof}

\subsection{Proof of Proposition 4.}
\begin{proof}
Fix an arbitrary realization of player $i$'s player-specific variable $\mathbf{Z}_{i} = \mathbf{z}_{i}$ and market-specific variable $\mathbf{X} = \mathbf{x}$.
By Assumption 3, the rival's player-specific variable $\mathbf{Z}_{j}$ has continuous variation conditional on $\mathbf{Z}_{i} = \mathbf{z}_{i}$.
I can find two realizations of $\mathbf{Z}_{j}$, and say $\mathbf{z}_{j}^{1}$ and $\mathbf{z}_{j}^{2}$.  
The equation \eqref{eq:semiparam-id} for these realizations leads to the following system of equations
\begin{align*}
    G_{i}^{-1}(P_{i}(\mathbf{x}, \mathbf{z}_{i}, \mathbf{z}_{j}^{1})) &= \pi_{i}(\mathbf{x}, \mathbf{z}_{i}) + \delta_{i}(\mathbf{x}, \mathbf{z}_{i}) P_{j}(\mathbf{x}, \mathbf{z}_{j}^{1}, \mathbf{z}_{i}) \\
    G_{i}^{-1}(P_{i}(\mathbf{x}, \mathbf{z}_{i}, \mathbf{z}_{j}^{2})) &= \pi_{i}(\mathbf{x}, \mathbf{z}_{i}) + \delta_{i}(\mathbf{x}, \mathbf{z}_{i}) P_{j}(\mathbf{x}, \mathbf{z}_{j}^{2}, \mathbf{z}_{i}) 
\end{align*}
By the Assumption 4, both values of $P_{i}(\mathbf{x}, \mathbf{z}_{i}, \mathbf{z}_{j}^{1})$ and $P_{i}(\mathbf{x}, \mathbf{z}_{i}, \mathbf{z}_{j}^{2})$ are known to the econometrician.
Moreover, by the Assumption 1, the function $G_{i}^{-1}(\cdot)$ can be derived by the econometrician.
The above system of equations contains two equations with two unknowns, $\pi_{i}(\mathbf{x}, \mathbf{z}_{i})$ and $\delta_{i}(\mathbf{x}, \mathbf{z}_{i})$.

Assumptions 3 and 4 imply that $P_{j}(\mathbf{x}, \mathbf{z}_{j}^{1}, \mathbf{z}_{i}) \neq P_{j}(\mathbf{x}, \mathbf{z}_{j}^{2}, \mathbf{z}_{i})$, and the rank condition for the above system of equations is satisfied.
Therefore, the unknown functions $\pi_{i}(\mathbf{x}, \mathbf{z}_{i})$ and $\delta_{i}(\mathbf{x}, \mathbf{z}_{i})$ are identified.
Since the choice of $\mathbf{z}_{i}$ and $\mathbf{x}$ was arbitrary, the functions $\pi_{i}(\mathbf{X}, \mathbf{Z}_{i})$ and $\delta_{i}(\mathbf{X}, \mathbf{Z}_{i})$ are identified for all values of $\mathbf{Z}_{i}$ and $\mathbf{X}$ in its support.

\end{proof}

\subsection{Proof of Proposition 5.}
\begin{proof}
For simplicity, suppose that $\pi_{i}(\mathbf{Z}_{i}) = \mathbf{Z}_{i}\beta_{i}$ and $\delta_{i}(\mathbf{Z}_{i}) = \delta_{i}$.
The parameter vector of interest is $\theta = (\beta_{1}, \beta_{2}, \delta_{1}, \delta_{2})$.
Assume that $Z_{ik}$ has an everywhere positive Lebesgue density conditional on $\mathbf{Z}_{i,-k}$ and the matrices $\mathbf{Z}_{i}$ and $\mathbf{Z}_{j}$ have full column rank.

Without loss of generality, assume $\beta_{2k} > 0$, and let $Z_{2k} $ be small enough conditional on $\mathbf{Z}_{2,-k}$.
Consequently, $\beta_{2} Z_{2k} \approx -\infty$, and player 2's expected deterministic payoff of choosing $Y_{2}=1$ is
\begin{equation*}
    EU_{2}(\mathbf{Z}_{2}, \mathbf{Z}_{1}; \theta) = \mathbf{Z}_{2}\beta_{2} + \delta_{2} \cdot P_{1}(\mathbf{Z}_{1}, \mathbf{Z}_{2}; \theta) \approx -\infty.
\end{equation*}
The conditional choice probability for player 2 is 
\begin{align*}
    P_{2}(\mathbf{Z}_{2},\mathbf{Z}_{1}; \theta) &= \int_{\varepsilon_{2}} \frac{P_{2}(\mathbf{Z}_{2},\mathbf{Z}_{1}; \theta) \exp\left( EU_{2}(\mathbf{Z}_{2}, \mathbf{Z}_{1}; \theta) + \varepsilon_{2}\right)^{1/\lambda_{2}}}{P_{2}(\mathbf{Z}_{2},\mathbf{Z}_{1}; \theta) \exp\left((EU_{2}(\mathbf{Z}_{2}, \mathbf{Z}_{1}; \theta) + \varepsilon_{2}) \right)^{1/\lambda_{2}} + (1 - P_{2}(\mathbf{Z}_{1},\mathbf{Z}_{2}; \theta))}  \, \mathrm{d}F_{2}(\varepsilon_{2}) \\
    &\approx \int_{\varepsilon_{2}} \frac{P_{2}(\mathbf{Z}_{2},\mathbf{Z}_{1}; \theta) \cdot 0}{P_{2}(\mathbf{Z}_{2},\mathbf{Z}_{1}; \theta) \cdot 0 + (1 - P_{2}(\mathbf{Z}_{2},\mathbf{Z}_{1}; \theta))}  \, \mathrm{d}F_{2}(\varepsilon_{2}) = 0
\end{align*}
as $EU_{2}(\mathbf{Z}_{2}, \mathbf{Z}_{1}; \theta) \approx -\infty$.
Thus, player 2 will choose the action $Y_{2}=0$ for all values of $Y_{1}$ and $\mathbf{Z}_{1}$.

Given that $P_{2}(\mathbf{Z}_{2},\mathbf{Z}_{1}; \theta) \approx 0$ as $EU_{2}(\mathbf{Z}_{2}, \mathbf{Z}_{1}; \theta) \approx -\infty$, player 1's expected deterministic payoff of choosing $Y_{1}=1$ is 
\begin{equation*}
    EU_{1}(\mathbf{Z}_{1}, \mathbf{Z}_{2}; \theta) = \mathbf{Z}_{1}\beta_{1} + \delta_{1} \cdot P_{2}(\mathbf{Z}_{2}, \mathbf{Z}_{1}; \theta) = \mathbf{Z}_{1}\beta_{1},
\end{equation*}
and the conditional choice probability is
\begin{equation*}
    P_{1}(\mathbf{Z}_{1},\mathbf{Z}_{2}; \theta) = \int_{\varepsilon_{1}} \frac{P_{1}(\mathbf{Z}_{1},\mathbf{Z}_{2}; \theta) \exp\left( \mathbf{Z}_{1}\beta_{1} + \varepsilon_{1} \right)^{1/\lambda_{1}}}{P_{1}(\mathbf{Z}_{1},\mathbf{Z}_{2}; \theta) \exp\left(\mathbf{Z}_{1}\beta_{1} + \varepsilon_{1} \right)^{1/\lambda_{1}} + (1 - P_{1}(\mathbf{Z}_{1},\mathbf{Z}_{2}; \theta))}  \, \mathrm{d}F_{1}(\varepsilon_{1})
\end{equation*}

Next, I claim that $P_{1}(1\,|\,\mathbf{Z}_{1},\mathbf{Z}_{2}; \theta) \neq P_{1}(1\,|\,\mathbf{Z}_{1},\mathbf{Z}_{2}; b)$ for $b \neq \theta$.

Suppose to the contrary $P_{1}(\mathbf{Z}_{1},\mathbf{Z}_{2}; \theta) = P_{1}(\mathbf{Z}_{1},\mathbf{Z}_{2}; b) \in (0,1)$.
The full rank condition on $\mathbf{Z}_{1}$ guarantees $\mathbf{Z}_{1}\beta_{1} \neq \mathbf{Z}_{1} b_{1}$ for $b_{1} \neq \beta_{1}$ and the exponential function is strictly monotonic.
It follows that 
\begin{equation*}
    \exp\left( \frac{1}{\lambda_{1}} (\mathbf{Z}_{1}\beta_{1}) \right) \neq \exp\left( \frac{1}{\lambda_{1}} (\mathbf{Z}_{1} b_{1}) \right).
\end{equation*}
Without loss of generality, assume $\mathbf{Z}_{1}\beta_{1} > \mathbf{Z}_{1} b_{1}$ for $b_{1} \neq \beta_{1}$ and $\exp\left( \frac{1}{\lambda_{1}} (\mathbf{Z}_{1}\beta_{1}) \right) > \exp\left( \frac{1}{\lambda_{1}} (\mathbf{Z}_{1} b_{1}) \right).$

Let the function $I(x)$ be defined by
\begin{equation*}
    I(x) = \frac{p \exp\left( \frac{1}{\lambda} (x + \varepsilon) \right)}{p \exp\left( \frac{1}{\lambda} (x + \varepsilon) \right) + (1 - p)}
\end{equation*}

This function is strictly increasing in $x$ since the derivative with respect to $x$ is always positive:
\begin{equation*}
    \frac{\mathrm{d} I}{\mathrm{d} x} = \frac{ \frac{1}{\lambda}p (1-p) \exp\left( \frac{1}{\lambda} (x + \varepsilon) \right)}{[p \exp\left( \frac{1}{\lambda} (x + \varepsilon) \right) + (1 - p)]^2}>0.
\end{equation*}

Thus I have
\begin{align*}
    P_{1}(\mathbf{Z}_{1},\mathbf{Z}_{2}; \theta) &= \int_{\varepsilon_{1}} \frac{P_{1}(\mathbf{Z}_{1},\mathbf{Z}_{2}; \theta) \exp\left( \frac{1}{\lambda_{1}} (\mathbf{Z}_{1}\beta_{1} + \varepsilon_{1}) \right)}{P_{1}(\mathbf{Z}_{1},\mathbf{Z}_{2}; \theta) \exp\left( \frac{1}{\lambda_{1}} (\mathbf{Z}_{1}\beta_{1} + \varepsilon_{1}) \right) + (1 - P_{1}(\mathbf{Z}_{1},\mathbf{Z}_{2}; \theta))}  \, \mathrm{d}F_{1}(\varepsilon_{1}) \\
    &= \int_{\varepsilon_{1}} \frac{P_{1}(\mathbf{Z}_{1},\mathbf{Z}_{2}; b) \exp\left( \frac{1}{\lambda_{1}} (\mathbf{Z}_{1} \beta_{1} + \varepsilon_{1}) \right)}{P_{1}(\mathbf{Z}_{1},\mathbf{Z}_{2}; b) \exp\left( \frac{1}{\lambda_{1}} (\mathbf{Z}_{1} \beta_{1} + \varepsilon_{1}) \right) + (1 - P_{1}(\mathbf{Z}_{1},\mathbf{Z}_{2}; b))}  \, \mathrm{d}F_{1}(\varepsilon_{1}) \\
    &> \int_{\varepsilon_{1}} \frac{P_{1}(\mathbf{Z}_{1},\mathbf{Z}_{2}; b) \exp\left( \frac{1}{\lambda_{1}} (\mathbf{Z}_{1} b_{1} + \varepsilon_{1}) \right)}{P_{1}(\mathbf{Z}_{1},\mathbf{Z}_{2}; b) \exp\left( \frac{1}{\lambda_{1}} (\mathbf{Z}_{1} b_{1} + \varepsilon_{1}) \right) + (1 - P_{1}(\mathbf{Z}_{1},\mathbf{Z}_{2}; b))}  \, \mathrm{d}F_{1}(\varepsilon_{1}) \\
    &= P_{1}(\mathbf{Z}_{1},\mathbf{Z}_{2}; b)
\end{align*}

Since I have assumed $\mathbf{Z}_{1}\beta_{1} > \mathbf{Z}_{1} b_{1}$ for $b_{1} \neq \beta_{1}$, it follows that $P_{1}(\mathbf{Z}_{1},\mathbf{Z}_{2}; \beta) > P_{1}(\mathbf{Z}_{1},\mathbf{Z}_{2}; b)$.
Therefore, the assumption that $P_{1}(\mathbf{Z}_{1},\mathbf{Z}_{2}; \beta) = P_{1}(\mathbf{Z}_{1},\mathbf{Z}_{2}; b)$ leads to a contradiction.
Hence, $P_{1}(\mathbf{Z}_{1},\mathbf{Z}_{2}; \beta) \neq P_{1}(\mathbf{Z}_{1},\mathbf{Z}_{2}; b)$ for $b \neq \beta$, and this result identifies $\beta_{1}$.

As for the identification of $\beta_{2}$, assume $\beta_{1k}>0$, and choose $Z_{1k}$ small enough such that $Z_{1k} \beta_{1k} \approx -\infty$.
$\beta_{2}$ is identified using the similar argument.

As for the identification of $\delta_{i}$ for $i=1,2$, let $Z_{ik} \approx \infty$.
$\delta_{1}$ and $\delta_{2}$ are identified using the similar argument.
\end{proof}


\end{appendix}

\end{document}